\documentclass[journal]{IEEEtran}

\usepackage{mathptmx}
\usepackage{amsthm,amsmath,amssymb}
\usepackage{bm} 
\usepackage{amsfonts}
\usepackage{subeqnarray}
\usepackage{cases,cite}
\usepackage[utf8]{inputenc}
\usepackage{graphicx}
\usepackage{amssymb}
\usepackage{float}
\usepackage{graphicx} 
\usepackage{subfigure}
\usepackage{algpseudocode}
\usepackage{url}
\usepackage{multirow}
\usepackage{microtype}
\usepackage{color}

\newtheorem{proposition}{Proposition}
\newtheorem{definition}{Definition}

\newtheorem{lemma}{Lemma}
\newtheorem{corollary}{Corollary}

\newtheorem{remark}{Remark}

\newtheorem{example}{Example}

\IEEEoverridecommandlockouts

\begin{document}

\title{{\texttt{FlipIn}: A Game-Theoretic Cyber Insurance Framework for Incentive-Compatible Cyber Risk Management of Internet of Things}}

\author{Rui~Zhang
        and~Quanyan~Zhu
\thanks{R. Zhang and Q. Zhu are with Department of Electrical and Computer Engineering, New York University, Brooklyn, NY, 11201
E-mail:\{rz885,qz494\}@nyu.edu. } }
\maketitle

\begin{abstract}
Internet of Things (IoT) is highly vulnerable to emerging Advanced Persistent Threats (APTs) that are often operated by well-resourced adversaries. Achieving perfect security for IoT networks is often cost-prohibitive if not impossible. Cyber insurance is a valuable mechanism to mitigate cyber risks for IoT systems. In this work, we propose a bi-level game-theoretic framework called \texttt{FlipIn} to design incentive-compatible and welfare-maximizing cyber insurance contracts. The framework captures the strategic interactions among APT attackers, IoT defenders, and cyber insurance insurers, and incorporates influence networks to assess the systemic cyber risks of interconnected IoT devices. The \texttt{FlipIn} framework formulates a game over networks within a principal-agent problem of moral-hazard type to design a cyber risk-aware insurance contract. We completely characterize the equilibrium solutions of the bi-level games for a network of distributed defenders and a semi-homogeneous centralized defender and show that the optimal insurance contracts cover half of the defenders' losses. Our framework predicts the risk compensation of defenders and the Peltzman effect of insurance. We study a centralized security management scenario and its decentralized counterpart, and leverage numerical experiments to show that network connectivity plays an important role in the security of the IoT devices and the insurability of both distributed and centralized defenders.
\end{abstract}

\begin{IEEEkeywords}
Cyber Insurance, Internet of Things, Game-Theoretic Design, \texttt{FlipIt} game, Influence Network, Principal-agent Problem, Moral Hazard, Information Asymmetry, Risk Compensation, Peltzman Effect, Network Effects 
\end{IEEEkeywords}

\section{Introduction}
\label{sec:Introduction}
The Internet of Things (IoT) has witnessed applications in many areas such as smart cities, healthcare, and transportations \cite{atzori2010internet, gubbi2013internet}. However, IoT networks and devices can be highly vulnerable to adversaries who can inflict huge financial and non-financial losses on government, companies, and nonprofit organization\cite{weber2010internet, suo2012security, zhao2013survey, jing2014security, sicari2015security}. For example, Mirai botnet in 2016 has compromised numerous IoT devices and knocked out popular sites, such as Netflix, Spotify, Twitter, and GitHub, with massive distributed denial-of-service (DDoS) attacks \cite{antonakakis2017understanding,kolias2017ddos}. 

\begin{figure}[]
\centering
{\includegraphics[width=0.45\textwidth]{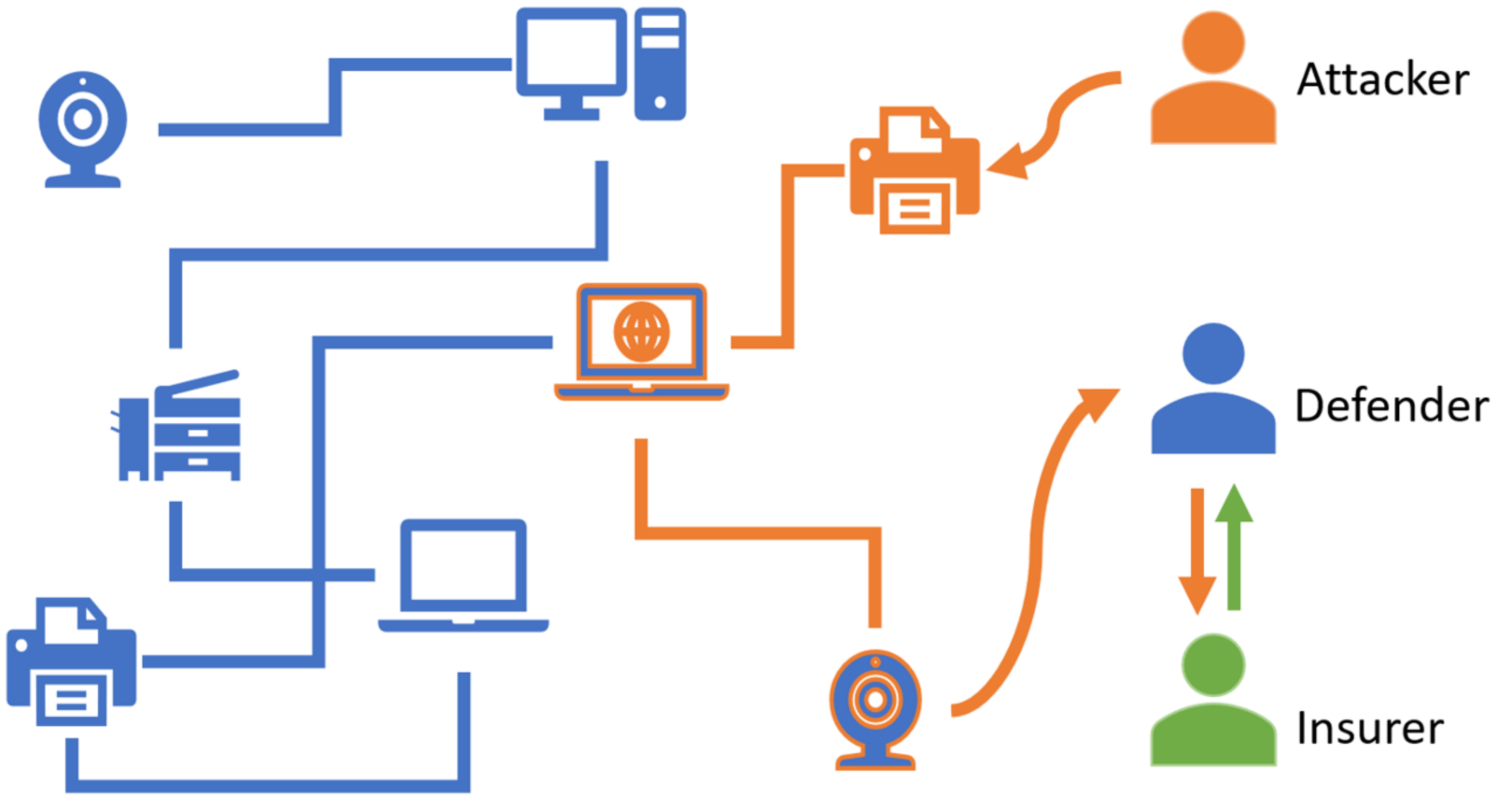}}
\caption{Cyber Insurance for IoT against APTs.}
\label{fig:F}
\end{figure}

One important class of sophisticated cyber-attacks called advanced persistent threats (APTs) has posed severe threats to IoT devices \cite{tankard2011advanced,cole2012advanced,abomhara2015cyber, hassanzadeh2015towards, hu2017defense}. Different from traditional cyber-attacks, APTs are executed by resourceful attackers, and they usually involve multiple steps and persist for a long period of time. For example, Stuxnet attack on Iran's nuclear program has compromised the target system's logic controllers, and then took control of the centrifuges, bringing them to failure \cite{langner2011stuxnet,kushner2013real}. 

The vulnerabilities of IoT devices to APTs arise from several aspects \cite{zhao2013survey,zhang2014iot,mahmoud2015internet}. First of all, security is not the primal concern during their design and the manufacturing. Users of the devices tend to adopt default or weak passwords. In addition, users do not maintain and patch their devices unless they stop working properly. Therefore, malicious activities of APTs are often unnoticed on IoT devices before they launch attacks and inflict significant losses. It is challenging to design effective defense methods and protect IoTs against APTs due to the coexistence of multiple types of devices and the lack of industrial standards. Moreover, the complexity of the IoT networks makes it difficult to investigate past cyber incidents. The defense solutions are also constrained by limited computational resources on IoT devices.

Cyber insurance becomes a new way to mitigate the risks from APTs to complement technological solutions and plays an important role when the technologies are imperfect \cite{majuca2006evolution, bohme2010modeling, marotta2017cyber}. An illustration of cyber insurance for IoT is provided in Fig. \ref{fig:F}. A defender first pays a premium to an insurer on his IoT devices, and when the devices are compromised by APTs, the insurer provides a financial coverage on various types of losses, such as data breaches, physical damages, and service shutdown. {The coverage and the financial protection against losses may prevent defenders from business discontinuities and provide them un-deprived resources to take actions to recover and defend themselves. }

Traditional insurance frameworks are not completely sufficient to address challenges of cyber insurance as the risks from APTs are caused by malicious attackers rather than accidents \cite{tankard2011advanced,cole2012advanced}. The design of effective cyber insurance contracts should take into account the attacker's behaviors and their impacts on the IoT systems. Moreover, the impact of APTs can propagate over IoT devices through network connections, and thus the insurers may bear extra risks if they fail to take into account the risk that arises from the interdependencies among IoT devices \cite{abomhara2015cyber, hassanzadeh2015towards, hu2017defense}. 

In this paper, we propose a bi-level game-theoretic framework called \texttt{FlipIn} to study the interactions among APT attackers, IoT defenders, and cyber insurance insurers over a network of IoT devices. In this framework, attackers and defenders compete over the ownership of IoT devices. The attackers aim to control longer periods so that they can conduct various malicious activities which inflict huge losses on the defenders. The defenders aim to reduce their losses by either controlling longer periods or purchasing cyber-insurance from insurers. With an objective of maximizing the revenue, the insurers charge premiums to the defenders and provide financial coverage to them when they face cyberattack-induced losses. 

As the impact of the attackers can propagate to other devices through network connections, there is a need to quantify the systemic risk of the whole network. To this end, we adopt linear influence models to capture the impact of one node on the others\cite{miura2008security,nguyen2009stochastic,alpcan2010network}. We further consider two scenarios of cyber insurance. The first one is a distributed scenario where each defender owns an IoT device in a network while the second one is its centralized counterpart where a centralized agent owns the network of nodes and plans the network defense. In both scenarios, there exists an attacker at each IoT device to compete with the defender over its ownership. 

\texttt{FlipIt} game has been broadly used to model the interactions between one APT attacker and one defender \cite{van2013flipit,bowers2012defending,pawlick2015flip,chen2017security}. We capture the adversarial interactions over networks by constructing local \texttt{FlipIt} games at each node in the distributed scenario and developing a global \texttt{FlipIt} game over a network in the centralized scenario. In the distributed
scenario, the local \texttt{FlipIt} games are composed into a network \texttt{FlipIt} game among multiple defenders and multiple attackers, which can be viewed as a bottom-up approach. In the centralized scenario, the proposed global \texttt{FlipIt} game can be decomposed into local \texttt{FlipIt} games at each node, which can be viewed as a top-down approach.

{Our \texttt{FlipIn} framework captures several unique features of IoT networks. Our framework does not require the networks to be homogeneous or fully connected. It captures IoT networks that are featured by various types of software, protocol, and hardware. Attackers or defenders may have different costs to attack or defend different IoT devices with unique physical components, respectively. In our framework, we capture those differences by the cost functions or parameters for the players to claim or reclaim the ownership of the IoT devices. Moreover, the failures of the IoT devices with sensitive information or crucial missions inflict huge losses to defenders while the failures of other devices may inflict less significant losses. Thus, defenders could encounter different losses on different devices, which is captured by the different loss parameters of defenders in our framework. Furthermore, IoT devices are often managed by different entities, and such decentralized ownership of IoT devices is investigated in the distributed scenario.}

Each defender interacts with an insurer, which is modeled by a class of moral-hazard type of principal-agent problems with incomplete information. The insurer acts as a principal and announces the insurance contract to the defender, while the defender acts as an agent and makes rational decisions. The incomplete information comes from the fact that the insurer cannot directly observe the interactions among defenders and attackers but can indirectly measure the defenders' losses as a consequence of their actions as well as the attackers' actions. 

The principal-agent problem and the network-based \texttt{FlipIt} games constitute a bi-level game among three parties, and the equilibrium solution to this composed game enables us to design effective cyber insurance contracts and mitigate financial impacts on the IoT networks and their operations.  We fully analyze the insurability of defenders by taking into account the individual rationalities for both defenders and insurers. We completely solve the insurance contract design problems in the distributed scenario and the semi-homogeneous centralized scenario. The optimal insurance contracts in both cases cover half of the defenders' losses when they are insurable. With numerical experiments, we show that network effects can damage the security of IoT networks and decrease the insurability of defenders. Our numerical experiments further indicate that nodes with more neighbors and networks with lower connectivities  are less insurable.

Our framework offers several insights on the best practices for IoT defenders on cyber risk management. Firstly, when the network influences are weak, defenders can successfully mitigate their risks through cyber insurance; when the network influences become stronger, defenders need to focus on deploying local protections instead of adopting cyber insurance. Secondly, the defenders who manage highly connected devices or sparsely connected networks do not benefit from cyber insurance. Thirdly, for weakly connected networks, decentralized management outperforms its centralized counterpart and each node is recommended to defend on its own while for strongly connected networks,  centralized management outperforms its decentralized counterpart and a global defender is preferred to be in charge of all devices.

\subsection{Organization of the Paper}
The rest of this paper is organized as follows. In Section \ref{sec:RelatedWorks}, we discuss the related works. Section \ref{sec:ProblemFormulation} formulates the problems and outlines the bi-level games for both distributed and centralized scenarios. {Section \ref{sec:FindEqu} provides an overview of finding the equilibrium.} Section \ref{sec:D} and Section \ref{sec:C} analyze the insurance contract design problems for two scenarios. Section \ref{sec:Num} presents numerical experiments and Section \ref{sec:Con} provides concluding remarks. A summary of notations has been provided in the following table.

\begin{table}[ht]
	\begin{tabular}{cc}
		\hline
		\multicolumn{2}{c}{Summary of Notations} 
		\\ \hline 
		\multicolumn{1}{c|}{$d$, $a$} & Defender, Attacker
		\\ \multicolumn{1}{c|}{$D$, $C$} & Distributed Scenario, Centralized Scenario
		\\ \multicolumn{1}{c|}{$\mathcal{N}$, $n$} & Set of Nodes, Node $n\in\mathcal{N}$
		\\ \multicolumn{1}{c|}{$p_{d,n}$, $p_{a,n}$} & {Defending Strategy/Frequency, Attacking Strategy/Frequency}
		\\ \multicolumn{1}{c|}{$\alpha_n$} & Expected Proportion of the Attacker's Controlling Time
		\\ \multicolumn{1}{c|}{$R_n$, $X_n$, $\beta_n$} & Defender's Risk, Direct Loss, Effective Loss
		\\ \multicolumn{1}{c|}{$w_{mn}$, $w_{mn}^*$} & Network Influence
		\\ \multicolumn{1}{c|}{$\eta$} & Discount Ratio of the Network Influence
		\\ \multicolumn{1}{c|}{$\gamma_{a,n}$, $c_{a,n}$} & Attacker's Utility Parameter, Cost Parameter
		\\ \multicolumn{1}{c|}{$\gamma_{d,n}$, $c_{d,n}$} & Defender's Loss Parameter, Cost Parameter
		\\ \multicolumn{1}{c|}{$s_n$, $s$} & Insurance Coverage level
		\\ \multicolumn{1}{c|}{$T_n$, $T$} & Insurance Premium
		\\ \hline
	\end{tabular}
\end{table}

\section{Related Works}
\label{sec:RelatedWorks}
{
Cyber insurance has been devised to mitigate cyber-risks by covering some of the losses caused by cyber-attacks \cite{majuca2006evolution, marotta2017cyber, bohme2010modeling}. Various frameworks have been proposed to study cyber-insurance from different perspectives. For example, Pal et al.,  have proposed a supply-demand model and showed that cyber insurance with client contract discrimination can improve network security \cite{pal2014will}; Khalili et al., have investigated the interdependent nature of cyber security and proposed a pre-screening method which is able to create profit opportunities for insurers \cite{khalili2017designing1, khalili2017designing2,khalili2018designing}; B{\"o}hme et al., have proposed several market models to study the information asymmetries between defenders and insurers as well as the interdependent security and correlated risks among defenders \cite{bohme2010modeling}; Vakilinia et al.,  have proposed a coalitional insurance framework where organizations insure a common platform instead of themselves under the consideration of their security interdependency \cite{vakilinia2018coalitional}. However, most of the current frameworks have not considered that the risks of cyber-attacks come from stealthy attackers with specific objectives and malicious activities, which is different from traditional insurance of mitigating risks from accidents. }

{
Game theory has been applied extensively to capture various types of cyber-attacks \cite{roy2010survey,ferdowsi2017colonel,han2019game}. For example, zero-sum games have been used to capture Jamming attacks and DoS/DDoS cyber-attacks\cite{li2015jamming,spyridopoulos2013gameC}; Stackelberg games have been applied to study moving target defense, honeypots, and correlated attacks \cite{vadlamudi2016moving,kiekintveld2015game,zhu2011stackelberg}; Signaling games have been used to investigate deception and data integrity attacks \cite{carroll2011game, teixeira2014security}; other applicable games have also been introduced to model different cyber-attack scenarios \cite{wang2014mean,zhu2015game}.  Game theory provides a theoretical analysis of cyber-attacks and offers valuable insights for cyber security administrators and managers to design detection and defense methods against them.}

{
APTs are severe threats to cyber security with the stealthiness and persistence nature \cite{tankard2011advanced,cole2012advanced,abomhara2015cyber, hassanzadeh2015towards, hu2017defense}.  Various papers have proposed different game-theoretic frameworks to investigate APTs and their impacts in different applications. In \cite{huang2018analysis}, Huang et al., have used a multi-stage Bayesian game to capture ATPs on cyber-physical systems; in \cite{hu2015dynamic}, Hu et al., have characterized the interplay among defenders, APT attackers, and insiders with a two-layer differential game; in \cite{xiao2017cloud}, Xiao et al., have applied prospect theory to investigate APTs on cloud storage systems; in \cite{min2018defense}, Min et al.  have captured the interactions between an APT attacker and a defender in a cloud storage system by a Colonel Blotto game; in \cite{rass2016gadapt}, Rass et al., have proposed a sequential game-theoretic framework to design defense strategies against APTs. }

{
The security aspect of IoT devices and their vulnerabilities to cyber-attacks have been reviewed and summarized by a lot of recent studies \cite{weber2010internet, suo2012security, zhao2013survey, jing2014security, sicari2015security}. Several game-theoretic frameworks have been proposed to investigate cyber-attacks towards IoT systems \cite{manshaei2013game,abuzainab2017dynamic,hu2019dynamic}. For example, Hamdi et al., have devised a game-theoretic model to study the adaptive security of eHealth IoT systems \cite{hamdi2014game}; Pouryazdan et al., have proposed a game-theoretic framework for trustworthy cloud-centric IoT applications \cite{pouryazdan2016game}; Namvar et al., have modeled the interaction between an IoT access point and a jammer with a Colonel Blotto game and proposed a centralized mechanism to address the jamming problem in the IoT systems composed of resource-constrained devices \cite{namvar2016jamming}.}

{
IoT systems could be extremely vulnerable to APTs because of complex types of devices and network connections, lack of detection and defense methods, and restricted computational resources \cite{abomhara2015cyber, hassanzadeh2015towards,zhao2013survey,zhang2014iot,mahmoud2015internet}. Game theory becomes a valuable tool to study APTs in IoT systems and design defense mechanisms against them. For example, in \cite{hu2017defense}, Hu et al., have proposed an expert system based APT detection game and showed its effectiveness to increase the security level of IoT systems; in \cite{lee2018game},  Lee et al., have proposed a game-theoretic vulnerability quantification method for social IoT systems against cyber-attacks including APTs. }

{
In this paper, we adopt \texttt{FlipIt} games to model APTs on IoT systems \cite{van2013flipit}. The \texttt{FlipIt} games have been used extensively to study APTs and investigate their impacts in various applications. In \cite{bowers2012defending}, Bowers et al., have demonstrated the application of \texttt{FlipIt} games to password reset policies, key rotation, VM refresh, and cloud auditing; in \cite{pawlick2015flip}, Pawlick et al., have used \texttt{FlipIt} game to model the competition between a defender and an attacker over the ownership of a cloud server;  in \cite{spyridopoulos2013gameJ}, Spyridopoulos et al., have proposed a variant of \texttt{FlipIt} game to study malware proliferation. These works have focused on the competition between one attacker and one defender over one resource. }

{
\texttt{FlipThem} games have been proposed to extend the \texttt{FlipIt} games and study the interactions between one defender and on attacker over multiple resources \cite{laszka2014flipthem, zhang2015game, leslie2015threshold, leslie2017multi}. In \cite{laszka2014flipthem}, Laszka et al., have proposed an AND control model where an attacker takes control of all resources to compromise the system and an OR model where an attacker only needs to take control of one resource to achieve that; in \cite{zhang2015game}, Zhang et al., have considered a situation where a defender and an attacker have strict constraints on their actions across all the resources; in \cite{leslie2015threshold} and \cite{leslie2017multi}, Leslie et al., have investigated a scenario where an attacker compromises a defender's resources with a threshold. However, they have not considered situations where multiple defenders and attackers interact in the same network of resources. Moreover, they have not captured the risk-dependencies between neighboring resources. }

{
Influence networks have been used to capture the interdependencies among neighboring nodes in cyber risk management \cite{alpcan2010network}. For example, in \cite{miura2008security}, Miura-Ko et al., have adopted influence networks to model interdependent security investments; in \cite{nguyen2009stochastic}, Nguyen et al., have presented an influence network model to study the vulnerability correlations of security assets. In this paper, we combine \texttt{FlipIt} games and influence networks and devise two scenarios of networked \texttt{FlipIt} games to capture the interactions between IoT defenders and APT attackers. }

{
We further capture the interactions between defenders and cyber insurance insurers by a moral-hazard type of principal-agent problem with incomplete information. Moral hazard has been discussed extensively in traditional insurance paradigm \cite{pauly1968economics, shavell1979moral}, and it has also been considered in cyber insurance \cite{gordon2003framework, bailey2014mitigating}. Principal-agent problems have also been applied to analyze cyber insurance \cite{pal2014improving,pal2017security}. The principal-agent problem and the nested \texttt{FlipIt} games constitute a bi-level game\cite{ghotbi2012bilevel, jenabi2013bi}. The game-of-games structure in bi-level games has been presented in various recent studies on cyber security \cite{zhu2015game, zhang2018game}. }

{
Compared with the bi-level game-theoretic cyber insurance framework proposed in \cite{zhang2017bi}, our framework captures the interactions between defenders and attackers with \texttt{FlipIt} games instead of zero-sum games to model APTs on IoT systems. One of our main contributions is that we extend the \texttt{FlipIt} game between one defender and one attacker to a network of \texttt{FlipIt} games by incorporating influence networks. We further consider both distributed and centralized scenarios of IoT defenders and analyze their interactions with both APT attackers and insurers. Our numerical experiments investigate the impacts of homogeneous/heterogeneous networks, homogeneous/heterogeneous players, and network connectivities to the security management of IoT defenders and the insurance contract designation of insurers, which have not been addressed in \cite{zhang2017bi}. The analysis offers several insights on the best practices for IoT defenders on cyber risk management and cyber insurance. Our framework further indicates that the optimal insurance coverage level is $1/2$ in the distributed scenario and the semi-homogeneous scenario and we show that an insurer can make a nonzero profit by providing cyber insurance. } 

\section{Problem Formulation}
\label{sec:ProblemFormulation}
{APTs are different from traditional cyber-attacks with the stealthiness and persistence nature. APT attackers have very specific objectives and compromise a system stealthily and slowly to maintain a small footprint and reduce detection risks. APTs, in general, could persist for long periods of time in a system. Moreover, APT attackers may prefer to stay anonymously and steal sensitive data or information instead of completely destroying services or physical components, which could inflict different types of financial and nonfinancial losses on defenders.    }

{We use \texttt{FlipIt} game to capture the competition between a defender and an APT attacker over the ownership of an IoT device \cite{van2013flipit}. In this game, a player takes control of the IoT device by moving with a cost and he only finds out about the state of the IoT device when he moves. This stealthy aspect of the game makes it suitable to capture APTs. The attacker has a specific objective to maximize his expected controlling time over the device. The attacker's malicious activities during the time that he controls the device could inflict various types of losses on the defender. The objective of the defender is to minimize his total expected losses caused the attacker. }

\begin{figure}[]
\centering
{\includegraphics[width=0.49\textwidth]{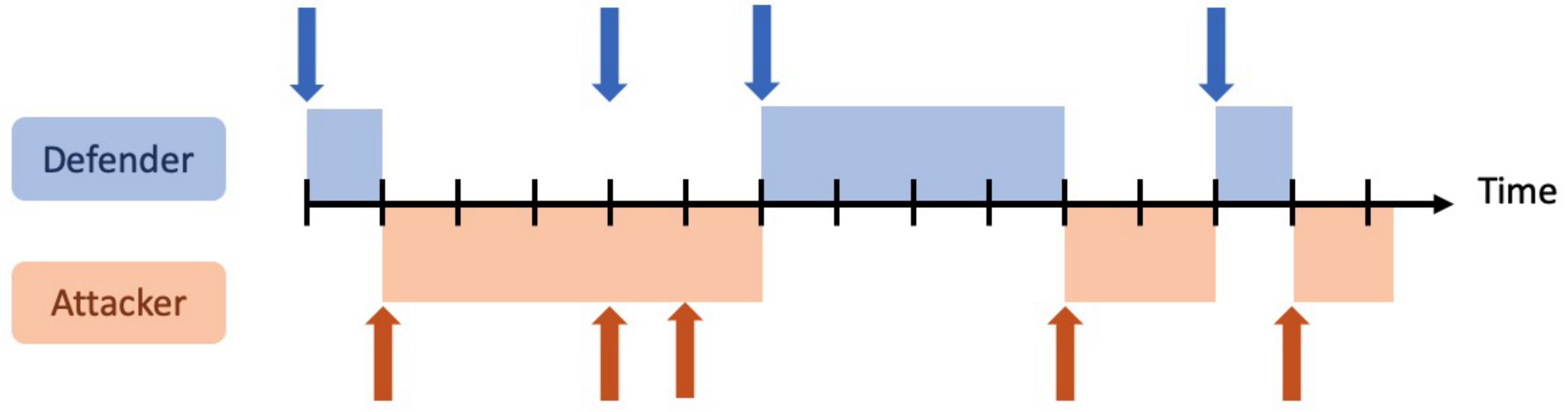}}
\caption{{The \texttt{FlipIt} game on Node $n$ between a defender and an attacker.}}
\label{fig:FlipIt}
\end{figure}

An illustration of the \texttt{FlipIt} game is provided in Fig. \ref{fig:FlipIt}: the game starts at time $0$ and the defender owns the device; at time $1$, the attacker attacks the device and then claims the ownership of it; at time $6$, the defender defends the device and reclaims the ownership of it; the game continues and the ownership of the device switches between the defender and the attacker. {Note that when the defender defends and the attacker attacks at the same time, e.g., time $4$, we follow the tie-breaking rule from \cite{van2013flipit} that their actions are ``canceled" and the ownership of the device does not change. Such a tie-breaking rule considers that the current owner of the device has the advantage to win the competition against the other player with prior knowledge or both players prefer not to move at the same time to avoid revealing themselves. This tie-breaking rule makes the game fully symmetric and allows us to handle ties smoothly.}

Consider an IoT network modeled by a directed graph $\mathcal{G}(\mathcal{N}, \mathcal{E})$ with $\mathcal{N}:=\{1,..., N\}$ and $\mathcal{E}$ denoting the set of nodes and the set of edges, respectively. Each node $n\in\mathcal{N}$ represents an IoT device which is controlled in turn by a defender and an APT attacker. Let us denote the sequence of the defender's defending time or the attacker's attacking time at node $n$ as an infinite non-decreasing sequence which can be specified as follows:
\begin{equation}
\label{eq:timeD}
\text{Defender:    } t_{d,n,1}, t_{d,n,2}, ..., t_{d,n,k-1}, t_{d,n,k}, t_{d,n,k+1},... 
\end{equation}
\begin{equation}
\label{eq:timeA}
\text{Attacker:    } t_{a,n,1}, t_{a,n,2}, ..., t_{a,n,k-1}, t_{a,n,k}, t_{a,n,k+1}, ... 
\end{equation}
Let $\phi_{d,n}(t_{d,n, k})$ or $\phi_{a,n}(t_{a,n,k})$ denote the feedback that the defender or the attacker at node $n$ receives when he defends or attacks at $t_{d,n,k}$ or $t_{a,n,k}$, where $\phi_{d,n}\in\Phi_{d,n}$ or $\phi_{a,n}\in\Phi_{a,n}$ with $\Phi_{d,n}$ or $\Phi_{a,n}$ denoting the set of all the possible forms of the feedback of the defender or the attacker, respectively. Some of the possible forms are listed here as examples \cite{van2013flipit}.
\begin{itemize}
\item (Non-adaptive) a player obtains no information, i.e., $\phi_{d,n}(t_{d,n,k}) = 0$ or $\phi_{a,n}(t_{a, n, k}) = 0$;
\item (Last move) a player obtains the opponent's last move time, i.e., $\phi_{d,n}(t_{d, n, k}) = \max \{t_{a, n, k'}|t_{a, n, k'} \leq t_{d, n, k}\}$ or  $\phi_{a,n}(t_{a, n, k}) = \max \{t_{d, n, k'}|t_{d, n, k'} \leq t_{a, n, k}\}$;
\item (Full history) a player obtains the full history of both players' moves, i.e., $\phi_{d,n}(t_{d, n, k}) = \left((t_{d, n, 1}, ..., t_{d, n, k}), (t_{a, n, 1}, ..., t_{a, n, k'})\right)$ or  $\phi_{a,n}(t_{a, n, k}) =  \left((t_{d, n, 1}, ..., t_{d, n, k'}), (t_{a, n, 1}, ..., t_{a, n, k})\right)$ where $t_{a, n, k'}= \max \{t_{a, n, k'}|t_{a, n, k'} \leq t_{d, n, k}\}$ or  $t_{d, n, k'} = \max \{t_{d, n, k'}|t_{d, n, k'} \leq t_{a, n, k}\}$.
\end{itemize}
A player decides the time of his next move following a strategy  based on his current time of move and the received feedback. Let $p_{d, n}:\mathbb{R}_{\geq 0}\times \Phi_{d,n}\rightarrow \mathbb{R}_{\geq 0}$ and $p_{a, n}:\mathbb{R}_{\geq 0}\times \Phi_{a,n}\rightarrow \mathbb{R}_{\geq 0}$ denote the defending strategy and the attacking strategy for the defender and the attacker at node $n$, respectively, we have
\begin{itemize}
\item when the defender defends at time $t_{d, n, k}$ and receives feedback $\phi_{d,n}(t_{d, n, k})$, his next defend will be at $t_{d,n,k+1} = p_{d,n}(t_{d,n,k}, \phi_{d,n}(t_{d, n, k}))$;
\item when the attacker attacks at time $t_{a, n, k}$ and receives feedback $\phi_{a,n}(t_{a, n, k})$, his next attack will be at  $t_{a,n,k+1} = p_{a,n}(t_{a,n,k}, \phi_{a,n}(t_{a, n, k}))$.
\end{itemize}

Let $\alpha_n \in [0, 1]$ denote the expected proportion of the time that the attacker controls node $n$ when the defender adopts a strategy of $p_{d, n}$ and the attacker adopts a strategy of $p_{a, n}$. While controlling the IoT devices, the attackers can benefit from conducting malicious activities, such as monitoring sensitive operations, stealing private information, and injecting ransomware. The objective of the attacker is to maximize the expected proportion of the time that he controls the device, which can be captured by the following optimization problem: 
\begin{equation}
\label{eq:AttackernMax}
\max\limits_{p_{a, n} \in \mathcal{S}_{a,n}} \  \gamma_{a, n} \alpha_n - c_{a, n}(p_{a, n}),
\end{equation}
where $\gamma_{a, n}\in\mathbb{R}_{\geq 0}$ denotes the utility parameter of the attacker and $\mathcal{S}_{a,n}=\{p_{a,n}\}$ denotes the set of all possible strategies by the attacker. Function $c_{a, n}:\mathcal{S}_{a,n}\rightarrow \mathbb{R}_{\geq 0}$ captures the cost of the attacker when he adopts a strategy of $p_{a, n}$.  Different $\gamma_{a, n}$ and $c_{a,n}$ capture the trade-offs between a larger proportion of controlling time and a smaller cost of the attacker. 

The attackers' activities can inflict huge financial and nonfinancial losses on defenders. A defender at one node may also face losses caused by the attackers at neighboring nodes, for example, the attackers can denial the services, send misleading information, or spread computer viruses to this node. Moreover, the impacts of attackers could propagate through network connections, such as computer worms and viruses, and thus, the defender may even face losses caused by undirectly connected nodes. 

{To measure the losses of defenders with respect to the attackers' controlling times, we first define the defender's risk level at node $n$ as
\begin{equation}
\label{eq:DefenderRBasic}
R_n = g_n(\alpha_1, \alpha_2, ..., \alpha_N),
\end{equation} 
where $g_n:([0,1])^N\rightarrow \mathbb{R}_{\geq 0}$ is a function of all the expected proportions of attackers' controlling times in this network.} We further assume that the defender's loss $X_n$ follows an exponential distribution of $\frac{1}{\gamma_{d,n}R_n}$ with its density function expressed as
\begin{equation}
\label{eq:Density}
\begin{array}{l}
f(x_n|\gamma_{d,n}R_n)=\frac{1}{\gamma_{d,n}R_n} \exp\left({-\frac{1}{\gamma_{d,n}R_n}x} \right),\forall x_n\in \mathbb{R}_{\geq 0}.
\end{array}
\end{equation} 
The parameter $\gamma_{d,n}\in\mathbb{R}_{>0}$ controls the level of the defender's losses, and a crucial IoT devices may have a larger $\gamma_{d, n}$. For example, an IoT device which monitors or controls nuclear reactors has a larger $\gamma_{d, n}$ than an IoT device which records the room temperature in a supermarket. {The exponential distribution has been widely applied to risk and insurance analysis \cite{dassios1989martingales, cizek2005statistical, balakrishnan2018exponential}.} We can see from the exponential distribution that the defender has a larger probability of facing a larger loss when he has a higher risk level. Furthermore, the expected loss satisfies that $\mathbb{E}(X_n) = \gamma_{d,n}R_{n}$, and the defender has a larger expected loss with the increase of his risk level. {An illustration of the relations between $p_{a,n}$, $p_{d,n}$, $\alpha_n$, $R_n$, and $X_n$ can be found in Fig. \ref{fig:D}.}
{\begin{remark}
\label{rem:LossNoNetwork}
When node $n$ has no network connections to any other nodes or there is only one node $n$ in this network, we could assume that $R_n =g(\alpha_1, \alpha_2, ..., \alpha_N) = \alpha_n$, which indicates that the defender at node $n$ has a higher risk level or faces a larger expected loss if the attacker at node $n$ has a larger expected proportion of controlling time.
\end{remark}}

\begin{figure}[]
\centering
{\includegraphics[width=0.49\textwidth]{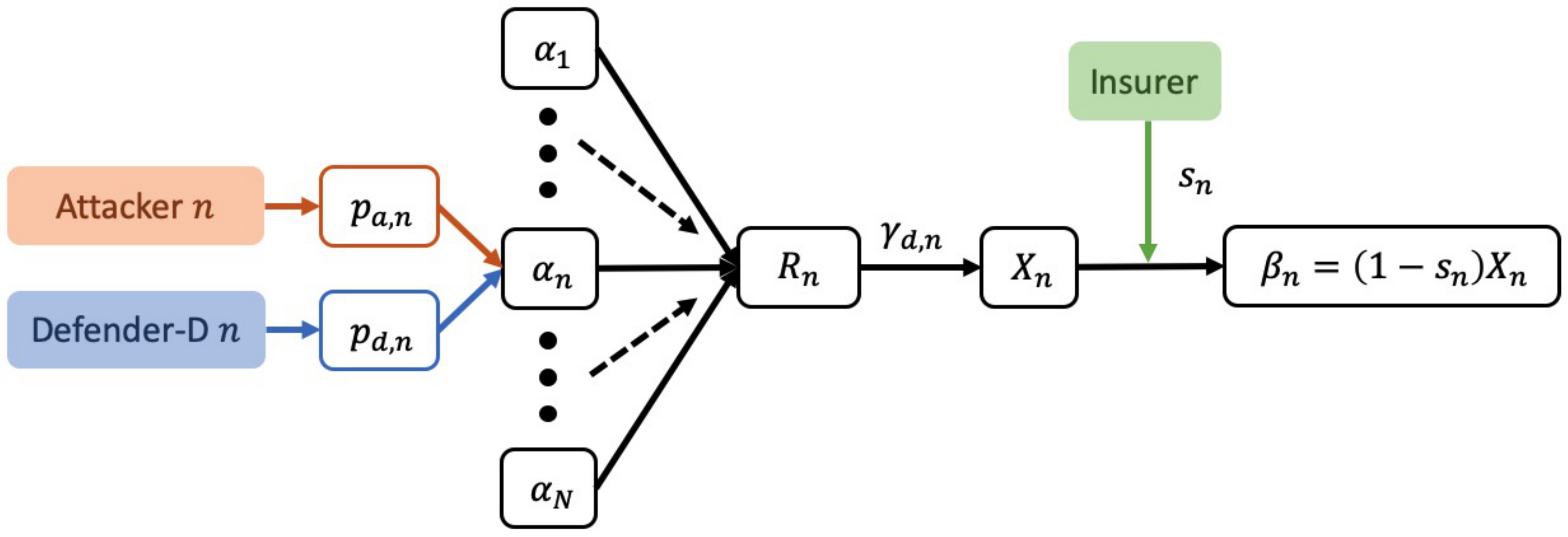}}
\caption{{Interactions among players at node $n$ in Defender-D.}}
\label{fig:D}
\end{figure}

The defenders can also purchase cyber-insurance to cover their losses. After paying a premium to the insurer, a defender receives a coverage from the insurer when he faces losses caused by attackers. In this paper, we consider two different scenarios of defenders, which are discussed separately in the following subsections. In summary, Defender-D represents a ``distributed" case where each defender at each node only considers his own losses, while Defender-C indicates a ``centralized" case where a global defender considers the overall loss of this IoT network.  {We present two examples to illustrate applications of Defender-D and Defender-C.}

{
\begin{example}[Centralized Scenario]
Consider a smart home which contains IoT devices, such as laptops, wireless routers, smart speakers, cameras, and sweeping robots. An APT attacker could compromise a smart speaker and record private conversions over days or months. The smart home owner could hardly notice the existence of the attacker as the smart speaker functions normally. The attacker could further leak the recorded conversions to the public or blackmail the owner. Similarly, APT attackers could record private videos from cameras or steal sensitive documents from laptops. The smart home owner could choose to insure his IoT devices with cyber insurance and the insurer would provide financial coverage to his losses from the leakage of private information by APT attackers. It is natural for the smart home owner to insure all the devices together, which indicates a centralized scenario.  
\end{example}}

{
\begin{example}[Distributed Scenario]
In vehicular applications and inter-networking technologies (VANET), a vehicle dynamically adjusts its route to destinations by communicating with different IoT devices, such as other vehicles, smart traffic lights, and phones owned by pedestrians. An APT attacker could compromise the vehicle and hijack the data sent and received by it. Different from traditional cyber-attacks, APTs aim not to immediately damage the vehicle as it can be detected by the security systems. With small modifications on the data over a long period of time, the vehicle may arrive at the desired destination by the attacker \cite{engoulou2014vanet,gantsou2015use,zhang2017strategic}. Moreover, an APT attacker could even affect the route of a target vehicle by sending misleading data from other vehicles and traffic lights. Cyber insurance could provide financial coverage to the passengers for their losses from arriving at a wrong destination. In this case, the security status of the vehicle is also affected by other IoT devices in this network. 
\end{example}}

\subsection{Defender-D}
In this case, there exist $N$ defenders in this IoT network and each IoT device is occupied by a defender. An illustration of the interactions between an attacker, a defender, and an insurer in one node is provided in Fig. \ref{fig:D}. Note that each attacker has no information about the players in other nodes and his objective is to maximize his expected proportion of controlling time $\alpha_n$ as in (\ref{eq:AttackernMax}). Each defender has no information about the actions of the players in other nodes, however, he may face losses caused by the attackers from other nodes.

After paying a premium to the insurer, the defender at node $n$ is entitled to receive a coverage $s_nX_n$ from the insurer when he faces a loss of $X_n$, where $s_n\in (0,1]$ denotes the coverage level of the insurance. The defender now has an effective loss of $\beta_n = X_n - s_nX_n = (1-s_n)X_n$. Note that $s_n = 0$ indicates that there is no insurance and the defender's effective loss $\beta_n$ equals to his direct loss $X_n$. As a result, the expected effective loss of the defender can be described as
\begin{equation}
\label{eq:DefendernEffectiveLoss}
\mathbb{E}[\beta_n] = \mathbb{E}[(1-s_n)X_n]=(1-s_n)\gamma_{d, n} R_n.
\end{equation}
{The objective of each defender is to minimize his expected effective loss, which can be captured by the following optimization problem after plugging (\ref{eq:DefenderRBasic}) into (\ref{eq:DefendernEffectiveLoss}):
\begin{equation}
\label{eq:DefendernMin}
\min\limits_{p_{d, n}\in \mathcal{S}_{d,n}} \  (1-s_n)\gamma_{d, n}g_n(\alpha_1, ...,\alpha_{N}) + c_{d, n}(p_{d, n}),
\end{equation}
where $\mathcal{S}_{d,n}=\{p_{d,n}\}$ denotes the set of all possible strategies by the defender. Function $c_{d, n}:\mathcal{S}_{d,n}\rightarrow \mathbb{R}_{\geq 0}$ captures the cost of the defender under the strategy $p_{d, n}$. The parameter $\gamma_{d, n}$ and function $c_{d, n}$ capture the trade-offs between a smaller expected effective loss and a higher cost. }

By comparing (\ref{eq:AttackernMax}) and (\ref{eq:DefendernMin}), we can see that the attacker aims to maximize the expected proportion of his controlling time while the defender aims to minimize it. The conflicting interest between the defender and the attacker constitutes a \texttt{FlipIt-D} game in each IoT device, and its Nash equilibrium is defined as follows.
\begin{definition}
\label{def:FlipItn}
Let $\mathcal{S}_{a, n} $ and $\mathcal{S}_{d, n}$ denote the strategy sets for the attacker and the defender at node $n$, respectively; let $J_{a,n}(p_{a,n},p_{d,n})$ and $J_{d,n}(p_{d,n},p_{a,n};s_n, \{\alpha_{m}\}_{m\neq n})$ denote the objective functions from (\ref{eq:AttackernMax}) and (\ref{eq:DefendernMin}), respectively. A strategy profile $\{p_{a,n}^*, p_{d,n}^*\}$ is a Nash equilibrium of the \texttt{FlipIt-D} game at node $n$ defined by $\langle \{\text{Attacker, Defender-D}\},  \{\mathcal{S}_{a,n},\mathcal{S}_{d,n}\}, \{J_{a,n}, J_{d,n}\}\rangle$ if 
\[ J_{a, n}(p_{a,n}^*, p_{d,n}^*) \geq J_{a, n}(p_{a,n}, p_{d,n}^*), \ \  \forall p_{a,n} \in \mathcal{S}_{a,n};\]
\[\begin{array}{l}
J_{d, n}(p_{d,n}^*, p_{a,n}^*;s_n,\{\alpha_{m}\}_{m\neq n}) \\ \ \ \ \ \ \ \ \ \ \  \leq J_{d, n}(p_{d,n}, p_{a,n}^*;s_n,\{\alpha_{m}\}_{m\neq n}), \ \  \forall p_{d,n} \in \mathcal{S}_{d,n}.
\end{array} \]
Furthermore, a strategy profile $\{\{p_{a,n}^*\}, \{p_{d,n}^*\}\}$ is a global Nash equilibrium of the \texttt{G-FlipIt-D} game defined by $\langle\{\text{Attackers, Defender-Ds}\},\{\{\mathcal{S}_{a,n}\},\{\mathcal{S}_{d,n}\}\}, \{\{J_{a,n}\}, \{J_{d,n}\}\}\rangle$ if
\[ J_{a, n}(p_{a,n}^*, p_{d,n}^*) \geq J_{a, n}(p_{a,n}, p_{d,n}^*), \ \  \forall p_{a,n} \in \mathcal{S}_{a,n}, n \in\mathcal{N};\]
\[\begin{array}{l}
J_{d, n}(p_{d,n}^*, p_{a,n}^*;s_n,\{\alpha_{m}^*\}_{m\neq n}) \\ \ \ \ \ \  \leq J_{d, n}(p_{d,n}, p_{a,n}^*;s_n,\{\alpha_{m}^*\}_{m\neq n}), \ \  \forall p_{d,n} \in \mathcal{S}_{d,n}, n \in\mathcal{N},
\end{array} \]
where $\alpha_m^*$ denotes the expected proportion of the time that the attacker controls node $m$ when the defender adopts a strategy of $p_{d, m}^*$ and the attacker adopts a strategy of $p_{a, m}^*$.
\end{definition}

\begin{figure}[]
\centering
{\includegraphics[width=0.3\textwidth]{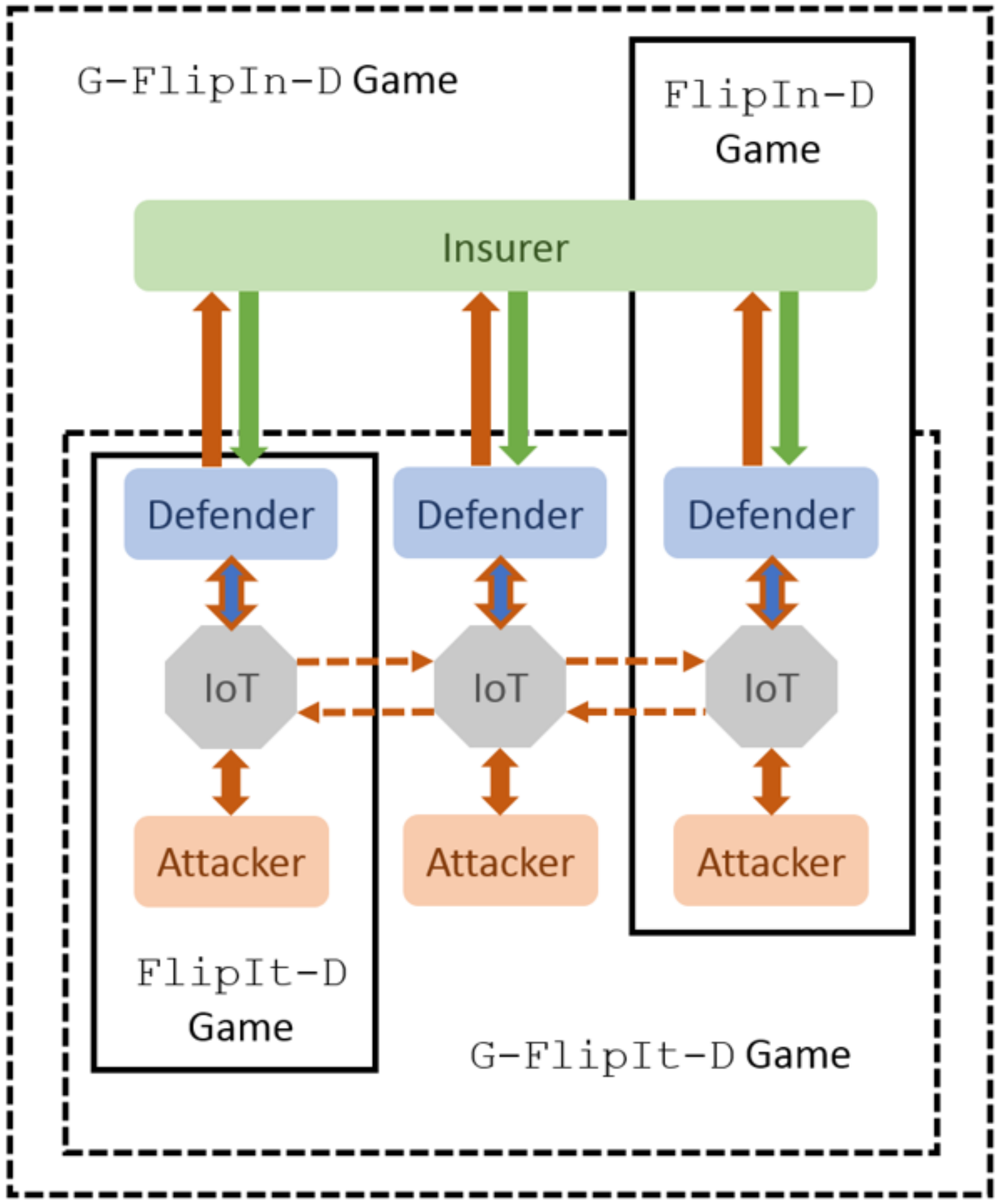}}
\caption{The structure of the games in Defender-D. There are $N$ defenders in this case. The \texttt{FlipIt-D} game is played between a defender and an attacker on one IoT device while the \texttt{G-FlipIt-D} game captures the interactions among all defenders and all attackers over the IoT network. The \texttt{FlipIn-D} game captures the interactions between a defender, an attacker and the insurer while the \texttt{G-FlipIn-D} game describes the interactions between all defenders, all attackers, and the insurer. }
\label{fig:Defendern}
\end{figure}

{An illustration of the \texttt{FlipIt-D} game and the \texttt{G-FlipIt-D} game is provided in Fig. \ref{fig:Defendern}. The Nash equilibrium of the \texttt{FlipIt-D} game is affected by the strategies of other players in the network, and it also affects the results of other \texttt{FlipIt-D} games.} The complex interactions among all attackers and defenders constitute a \texttt{G-FlipIt-D} game, whose Nash equilibrium is achieved when all the \texttt{FlipIt-D} games in each node reach their Nash equilibriums. Note that the Nash equilibrium of the \texttt{FlipIt-D} game is affected by the coverage level $s_n$ through $J_{d, n}$, and the Nash equilibrium of the \texttt{G-FlipIt-D} game is affected by the coverage levels $\{s_n\}$ of all nodes.  

The defender at node $n$ has to pay a premium $T_n\in\mathbb{R}_{\geq 0}$ to receive a coverage of his losses. A rational defender purchases cyber insurance only when his total loss including the premium under the insurance is lower than his total loss without the insurance, which can be expressed as the following individual rationality constraint for the defender:
{
\begin{equation}
\label{eq:DefendernIR}
\begin{array}{l}
(1-s_n)\gamma_{d, n} g_n\left(\alpha_1, ..., \alpha_n^*(s_n), ..., \alpha_N\right) + c_{d, n}\left(p_{d, n}^*(s_n)\right)  + T_n \\ \ \ \ \ \ \ \ \ \ \ \leq \gamma_{d,n}g_n\left(\alpha_1', ..., \alpha_n^*(0), ..., \alpha_N'\right)+ c_{d, n}\left(p_{d, n}^*(0)\right).
\end{array}
\end{equation}}
Note that we have abused the notation with the purpose of simplifying illustration: $\alpha_n^*(s_n)$ and $\alpha_n^*(0)$ denote the equilibrium expected proportions of the attacker's controlling time at node $n$ with the insurance of coverage level $s_n$ and without the insurance, respectively; $p_{d,n}^*(s_n)$ and $p_{d,n}^*(0)$ denote the equilibrium {defending strategies} of the defender at node $n$  with the insurance of coverage level $s_n$ and without the insurance, respectively. Specially, $\alpha_m$ and $\alpha_m'$ represent the expected proportions at node $m \neq n$ given the coverage level $s_n$ and $s_n=0$ at node $n$, respectively. 

{The insurer receives a premium $T_n$ from the defender and provides coverage $s_n\gamma_{d, n}g_n\left(\alpha_1, ..., \alpha_n^*(s_n), ..., \alpha_N\right)$ to him. Thus, the insurer has a profit of $T_n-s_n\gamma_{d, n}g_n\left(\alpha_1, ..., \alpha_n^*(s_n), ..., \alpha_N\right)$.} The insurer insures the defender only when he can make a profit, which can be expressed as the following individual rationality constraint for the insurer:
{
\begin{equation}
\label{eq:InsurernIR}
\begin{array}{l}
T_n-s_n\gamma_{d, n}g_n\left(\alpha_1, ..., \alpha_n^*(s_n), ..., \alpha_N\right) \geq 0.
\end{array}
\end{equation}}

The objective of the insurer is to make a larger profit from providing the insurance, which can be captured as the following optimization problem:
{
\begin{equation}
\label{eq:Insurern}
\begin{array}{l}
\max\limits_{ \{s_n,T_n\}  } \  T_n-s_n\gamma_{d, n}g_n\left(\alpha_1, ..., \alpha_n^*(s_n), ..., \alpha_N\right)\\
\begin{array}{cc}
{\text{s.t.}}&{\begin{array}{c}
(\ref{eq:DefendernIR}), (\ref{eq:InsurernIR}). 
\end{array}   }
\end{array}
\end{array}
\end{equation}}
We can see that the solution of (\ref{eq:Insurern}) depends on both the outcome of the \texttt{FlipIt-D} game at node $n$ and the expected proportions of the attackers' controlling times at other nodes. An effective insurance contract must satisfy both (\ref{eq:DefendernIR}) and (\ref{eq:InsurernIR}). The interactions between the defender and the insurer constitute a moral-hazard type of principal-agent problem where the insurer as a principal announces the insurance contract and the defender as an agent makes rational decisions on purchasing the insurance. Note that the insurer has incomplete information about the defender as he decides the insurance contract based on the outcome of the \texttt{FlipIt-D} game. 

The complex interactions between an attacker, a defender, and an insurer constitute a bi-level \texttt{FlipIn-D} game whose Nash equilibrium is defined as follows.
\begin{definition}
\label{def:BiLeveln}
{Let $\mathcal{S}_{i, n} = \{ \{s_n, T_n\} | s_n \in (0, 1], T_n \in \mathbb{R}_{\geq 0}, (\ref{eq:DefendernIR}), (\ref{eq:InsurernIR})\}$ denote the action set for the insurer; let $J_{i, n}(s_n, T_n)$ denote the objective function from (\ref{eq:Insurern}). Recall $\mathcal{S}_{d,n}$, $\mathcal{S}_{a,n}$, $J_{d,n}$, $J_{a,n}$ from Definition \ref{def:FlipItn}, a strategy profile $\{p_{a,n}^*, p_{d,n}^*, \{s_n^*, T_n^*\}\}$ is a Nash equilibrium of the bi-level \texttt{FlipIn-D} game defined by $\left\langle\{\text{Attacker, Defender-D, Insurer}\}, \{ \mathcal{S}_{a,n},\mathcal{S}_{d,n}, \mathcal{S}_{i, n}\},\right.$ $\left.\{J_{a,n}, J_{d,n}, J_{i, n}\}\right\rangle$ if $\{s_n^*, T_n^*\}$ solves (\ref{eq:Insurern}) and $\{p_{a,n}^*, p_{d,n}^*\}$ is a Nash equilibrium of the \texttt{FlipIt-D} game defined in Definition \ref{def:FlipItn} under $\{s_n^*, T_n^*\}$.}

{Furthermore, a strategy profile $\{\{p_{a,n}^*\}, \{p_{d,n}^*\}, \{\{s_n^*\}, \{T_n^*\}\}\}$ is a global Nash equilibrium of the bi-level \texttt{G-FlipIn-D} game defined by $\left\langle\{\text{Attackers, Defender-Ds, Insurer}\}  \{\{\mathcal{S}_{a,n}\},\{\mathcal{S}_{d,n}\}, \{\mathcal{S}_{i, n}\}\},\right.$ $\left. \{\{J_{a,n}\}, \{J_{d,n}\}, \{J_{i, n}\}\}\right\rangle$ if $\{s_n^*, T_n^*\}$ solves (\ref{eq:Insurern}) for all $n\in\mathcal{N}$ and $\{\{p_{a,n}^*\}, \{p_{d,n}^*\}\}$ is a Nash equilibrium of the \texttt{G-FlipIt-D} game defined in Definition \ref{def:FlipItn} under $\{\{s_n^*\}, \{T_n^*\}\}$.}
\end{definition}
An illustration of the \texttt{FlipIn-D} game and the \texttt{G-FlipIn-D} game has been provided in Fig. \ref{fig:Defendern}. The Nash equilibrium of the \texttt{FlipIn-D} game is affected by the other \texttt{FlipIn-D} games through $\alpha_m$ in (\ref{eq:Insurern}). The complex interactions among all players constitute a \texttt{G-FlipIn-D} game, whose Nash equilibrium is achieved when all the \texttt{FlipIn-D} games reach their Nash equilibriums.

\subsection{Defender-C}
In this case, there exists only one global defender in this IoT network. The interactions between the defender, attackers, and an insurer are illustrated in Fig. \ref{fig:C}. Similar to the case of Defender-D, each attacker has no information about the players in other nodes. The insurer offers one insurance contract $\{s,T\}$ to the defender which covers the losses of all the IoT devices. The defender's expected effective loss $\beta$ in this case can be expressed as
{
\begin{equation}
\label{eq:DefenderNEffectiveLoss}
\begin{array}{l}
\mathbb{E}(\beta) = \mathbb{E}\left((1-s)\sum\limits_{n=1}^N X_n\right) = (1-s)\sum\limits_{n=1}^N\gamma_{d,n}g_n\left(\alpha_1, ..., \alpha_N\right).
\end{array}
\end{equation}}

The defender aims to minimize his overall expected effective losses, which can be captured by the following optimization problem:
{
\begin{equation}
\label{eq:DefenderNMin}
\min\limits_{\{p_{d, n}\}}\ (1-s)\sum\limits_{n=1}^N \gamma_{d,n}g_n\left(\alpha_1, ..., \alpha_N\right) + \sum\limits_{n=1}^Nc_{d, n}\left(p_{d, n}\right) .
\end{equation}}

\begin{figure}[]
\centering
{\includegraphics[width=0.49\textwidth]{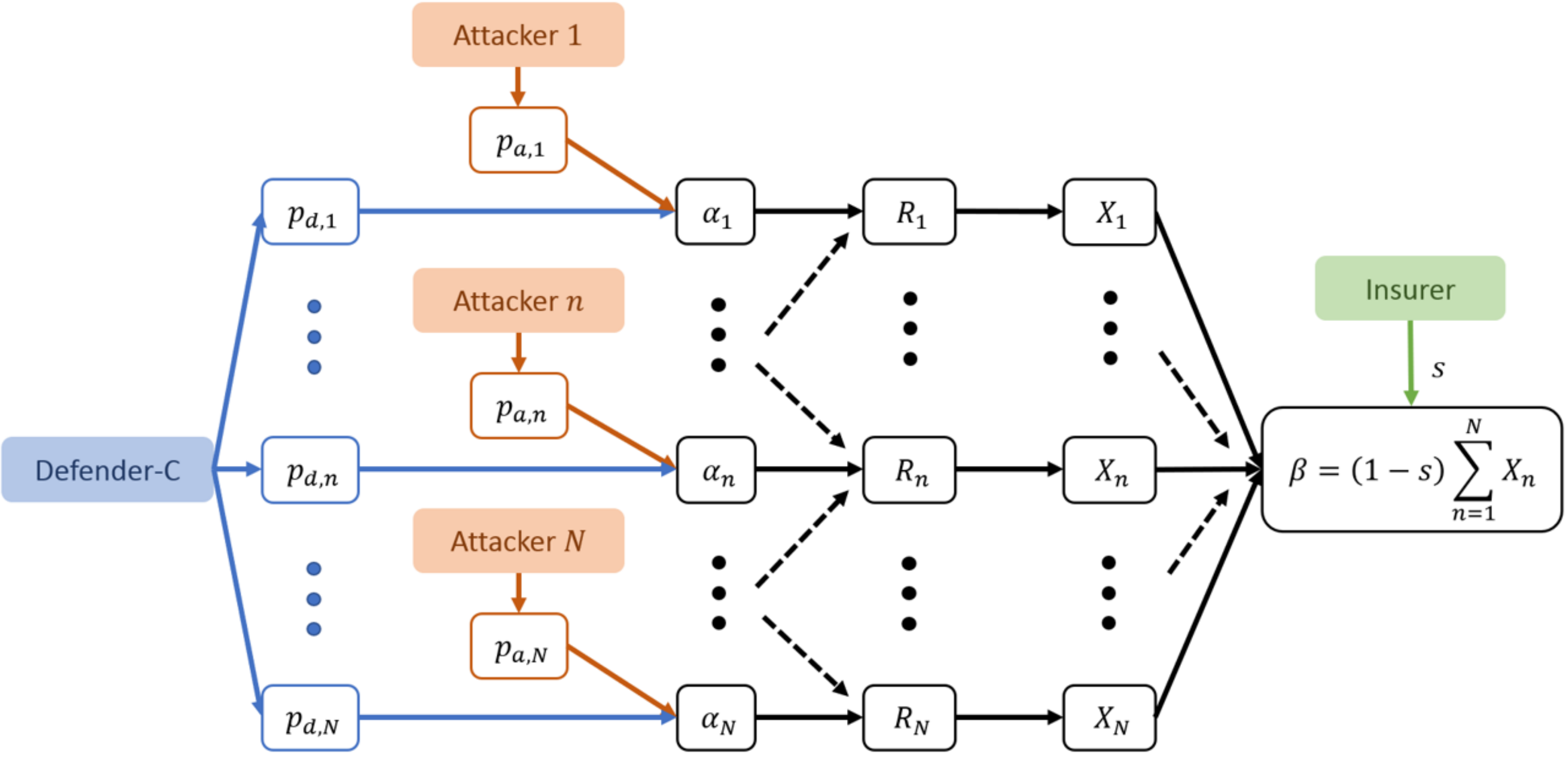}}
\caption{Interactions among players in the IoT network in Defender-C. }
\label{fig:C}
\end{figure}

The interactions between the defender and $N$ attackers in this IoT network constitute a global \texttt{FlipIt-C} game, which is defined as follows.
{
\begin{definition}
\label{def:FlipItN}
Let $\mathcal{S}_{a, n} $ and $\mathcal{S}_{d} $ denote the strategy sets for the attacker and the defender, respectively; let $J_{a,n}(p_{a,n}, p_{d, n})$ and $J_{d}(\{p_{d,n}\},\{p_{a,n}\};s)$ denote the objective functions from (\ref{eq:AttackernMax}) and (\ref{eq:DefenderNMin}), respectively. 
A strategy pair $\{\{p_{a,n}^*\}, \{p_{d,n}^*\}\}$ is a Nash equilibrium of the \texttt{FlipIt-C} game defined by $\left\langle\{\text{Attackers, Defender-C}\}, \{\{\mathcal{S}_{a,n}\},\mathcal{S}_{d}\}, \{\{J_{a,n}\}, J_{d}\}\right\rangle)$ if 
\[ J_{a, n}(p_{a,n}^*, p_{d,n}^*) \geq J_{a, n}(p_{a,n}, p_{d,n}^*), \ \  \forall p_{a,n} \in \mathcal{S}_{a,n}, n\in\mathcal{N};\]
\[ J_{d}(\{p_{d,n}^*\},\{p_{a,n}^*\};s) \leq J_{d}(\{p_{d,n}\},\{p_{a,n}^*\};s), \ \  \forall \{p_{d,n}\} \in \mathcal{S}_{d}.\]
\end{definition}}
Note that the Nash equilibrium of the \texttt{FlipIt-C} game is affected by the coverage level $s$ through $J_{d}$. {In the previous subsection, we have shown that $N$ \texttt{FlipIt-D} games constitute a global \texttt{G-FlipIt-D} game, while in this subsection, we show that the \texttt{FlipIt-C} game may be decentralized into $N$ local \texttt{L-FlipIt-C} games under certain conditions. }
{
\begin{remark}[Decentralization]
\label{rem:DefenderNMin}
If $g_n(\alpha_1, ..., \alpha_N)$ is additively separable for all $1 \leq n \leq N$, i.e., $g_n(\alpha_1, ..., \alpha_N)=\sum_{m=1}^Ng_{n,m}(\alpha_m)$, we have $\sum_{n=1}^N\gamma_{d,n} g_n\left(\alpha_1, ..., \alpha_N\right) =\sum_{n=1}^N\sum_{m=1}^N \gamma_{d,n} g_{n,m}\left(\alpha_m\right)=\sum_{n=1}^N\sum_{m=1}^N \gamma_{d,m} g_{m,n}\left(\alpha_n\right)$. Thus, solving the global problem (\ref{eq:DefenderNMin}) is equivalent to solving the following $N$ sub-problems at each node:
\begin{equation}
\label{eq:DefenderNMinSub}
\min\limits_{p_{d, n}}\ (1-s)\sum_{m=1}^N \gamma_{d,m} g_{m,n}\left(\alpha_n\right) +c_{d, n}\left(p_{d, n}\right) .
\end{equation}
\end{remark}
With Remark \ref{rem:DefenderNMin}, the interactions between the attacker and the defender at node $n$ constitute a \texttt{L-FlipIt-C} game, which is defined as follows.}
{
\begin{definition}
\label{def:FlipItNn}
Let $\mathcal{S}_{a, n} $ and $\mathcal{S}_{d, n} $ denote the strategy sets for the attacker and the defender at node $n$, respectively; let $J_{a,n}(p_{a,n},p_{d,n})$ and $J_{d,n}(p_{d,n},p_{a,n};s)$ denote the objective functions from (\ref{eq:AttackernMax}) and (\ref{eq:DefenderNMinSub}), respectively. 
A strategy profile $\{p_{a,n}^*, p_{d,n}^*\}$ is a Nash equilibrium of the \texttt{L-FlipIt-C} game at node $n$ defined by $\left\langle\{\text{Attacker, Defender-C}\},\{\mathcal{S}_{a,n},\mathcal{S}_{d,n}\}, \{J_{a,n}, J_{d,n}\}\right\rangle$ if 
\[ J_{a, n}(p_{a,n}^*, p_{d,n}^*) \geq J_{a, n}(p_{a,n}, p_{d,n}^*), \ \  \forall p_{a,n} \in \mathcal{S}_{a,n};\]
\[\begin{array}{l}
J_{d, n}(p_{d,n}^*, p_{a,n}^*;s) \leq J_{d, n}(p_{d,n}, p_{a,n}^*;s), \ \  \forall p_{d,n} \in \mathcal{S}_{d,n}.
\end{array} \]
\end{definition}
An illustration of the \texttt{FlipIt-C} game and the \texttt{L-FlipIt-C} games is provided in Fig. \ref{fig:DefenderN}. We can see that all the \texttt{L-FlipIt-C} games are independent of each other as (\ref{eq:AttackernMax}) and (\ref{eq:DefenderNMinSub}) do not depend on the outcomes of other \texttt{L-FlipIt-C} games. However, the \texttt{L-FlipIt-C} game at node $n$ takes the parameters $\{\gamma_m\}$ from other nodes into consideration. Note that when $g_n(\alpha_1, ..., \alpha_N)$ is not additively separable, we cannot decentralize the \texttt{FlipIt-C} game and obtain the \texttt{L-FlipIt-C} games. }

\begin{figure}[]
\centering
{\includegraphics[width=0.3\textwidth]{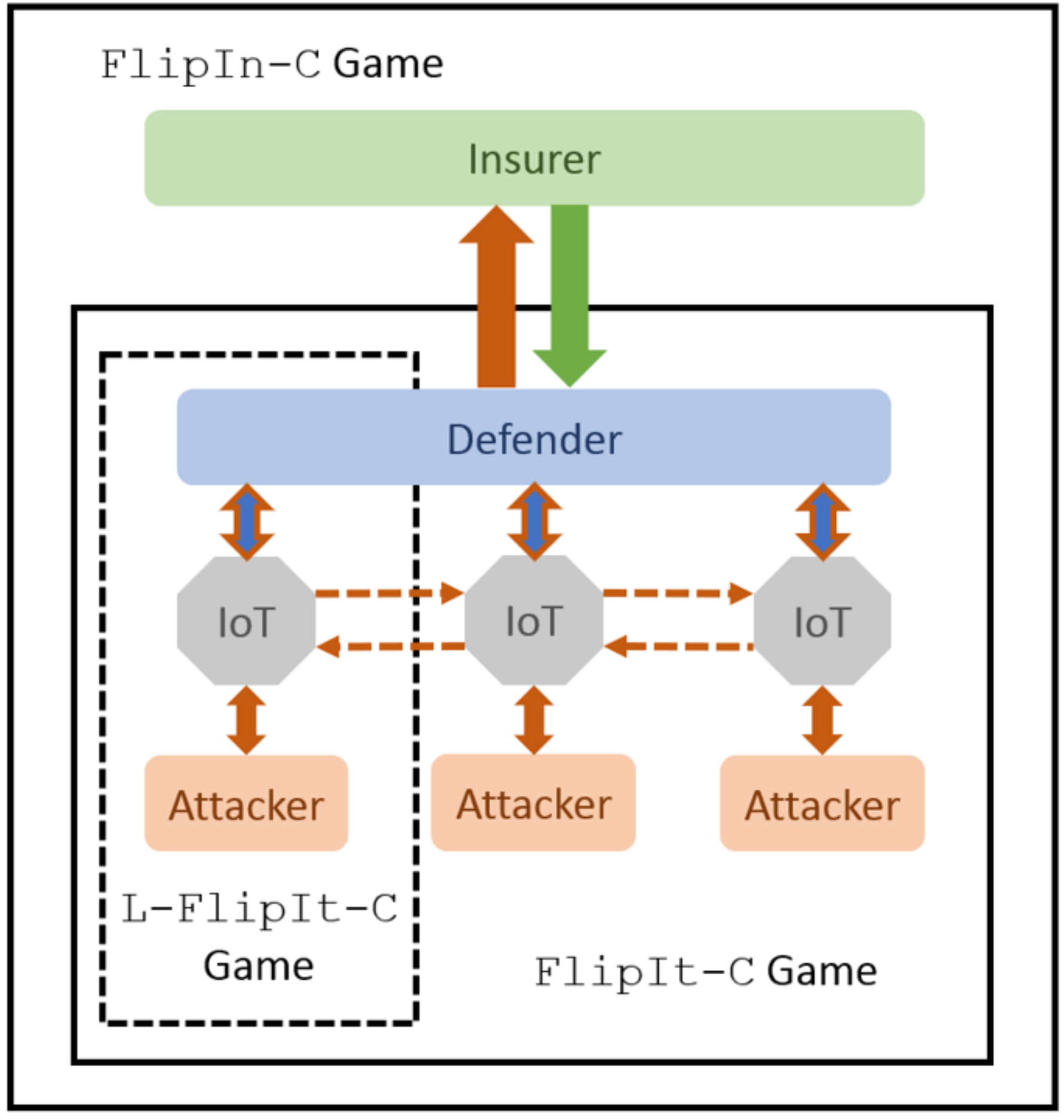}}
\caption{The structure of the games in Defender-C. There is only one defender in this case. The \texttt{FlipIt-C} game captures the interactions between a defender and all attackers in the IoT network while the \texttt{L-FlipIt-C} game captures the interactions between the defender and an attacker on one IoT device. The \texttt{FlipIn-C} captures the interactions between the defender, all attackers, and the insurer. }
\label{fig:DefenderN}
\end{figure}

The defender pays a premium $T\in\mathbb{R}_{\geq 0}$ to insure the IoT network and receive a coverage when he faces losses from any devices. Following similar steps in Section \ref{sec:ProblemFormulation}.A., the defender's individual rationality constraint in this case can be written as
{
\begin{equation}
\label{eq:DefenderNIR}
\begin{array}{l}
(1-s)\sum\limits_{n=1}^N \gamma_{d, n} g_n\left(\alpha_1^*(s), ..., \alpha_N^*(s)\right)  + \sum\limits_{n=1}^Nc_{d, n}\left(p_{d, n}^*(s)\right)   + T \\ \ \ \ \ \ \ \ \ \ \ \leq \sum\limits_{n=1}^N\gamma_{d, n} g_n\left(\alpha_1^*(0), ..., \alpha_N^*(0)\right) + \sum\limits_{n=1}^Nc_{d, n}\left(p_{d, n}^*(0)\right).
\end{array}
\end{equation}}
The insurer's individual rationality constraint can be written as
{
\begin{equation}
\label{eq:InsurerNIR}
\begin{array}{l}
T-s\sum\limits_{n=1}^N g_n\left(\alpha_1^*(s), ..., \alpha_N^*(s)\right) \geq 0.
\end{array}
\end{equation}}

The insurer aims to maximize the profit as follows:
{
\begin{equation}
\label{eq:InsurerN}
\begin{array}{l}
\max\limits_{ \{s,T\}  } \ T-s\sum\limits_{n=1}^N\gamma_{d, n} g_n\left(\alpha_1^*(s), ..., \alpha_N^*(s)\right)\\
\begin{array}{cc}
{\text{s.t.}}&{\begin{array}{c}
(\ref{eq:DefenderNIR}), (\ref{eq:InsurerNIR}). 
\end{array}   }
\end{array}
\end{array}
\end{equation}}

Similar to the case in the previous subsection, the interactions between the defender and the insurer constitute a principal-agent problem with incomplete information; the complex interactions between the defender, the attackers, and the insurer constitute a bi-level \texttt{FlipIn-C} game whose Nash equilibrium is defined as follows. 
{
\begin{definition}[Equilibrium Concept for \texttt{FlipIn-C}]
\label{def:BiLevelN}
Let $\mathcal{S}_{i} = \{ \{s, T\} | s \in (0, 1], T \in \mathbb{R}_{\geq 0}, (\ref{eq:DefenderNIR}), (\ref{eq:InsurerNIR})\}$ denote the action set for the insurer; let $J_{i}(s, T)$ denote the objective function from (\ref{eq:InsurerN}). Recall $\mathcal{S}_{d}$, $\mathcal{S}_{a,n}$, $J_{d}$, $J_{a,n}$ from Definition \ref{def:FlipItN}, a strategy profile $\{\{p_{a,n}^*\}, \{p_{d,n}^*\}, \{s^*, T^*\}\}$ is a global Nash equilibrium of the bi-level \texttt{FlipIn-C} game defined by $\left\langle\{\text{Attackers, Defender-C, Insurer}\}, \{\{\mathcal{S}_{a,n}\},\mathcal{S}_{d}, \mathcal{S}_{i}\},\right. $ $\left. \{\{J_{a,n}\}, J_{d}, J_{i}\}\right\rangle$ if $\{s^*, T^*\}$ solves (\ref{eq:InsurerN}) and $\{p_{a,n}^*, p_{d,n}^*\}$ is a Nash equilibrium of the \texttt{FlipIt-C} game defined in Definition \ref{def:FlipItN} under $\{s^*, T^*\}$.
\end{definition}}
An illustration of the \texttt{FlipIn-C} game has been provided in Fig. \ref{fig:DefenderN}. Compared with the \texttt{G-FlipIn-D} game, the \texttt{FlipIn-C} game contains only one principal-agent problem and there is only one defender who competes with each attacker on the ownership of each IoT device. Moreover, the \texttt{G-FlipIn-D} game can be considered as a bottom-up approach on a distributed scenario as the distributed games constitute a centralized game, while the \texttt{FlipIn-C} game can be considered as a top-down approach on a centralized scenario as the centralized \texttt{FlipIt-C} game can be decentralized into distributed \texttt{L-FlipIt-C} games. 

{
\section{Overview of Finding the Equilibrium}
\label{sec:FindEqu}
It is challenging to directly compute the equilibrium of the bi-level \texttt{FlipIn} games in both the distributed case and the centralized case due to the complex relations among three networked parties of players and the various strategy choices of defenders and attackers. In this paper, we consider that both defenders and attackers adopt non-adaptive periodic strategies. The periodic strategy could be viewed as a routine security examination of the defender or programmed attacks of the attacker. Moreover, we incorporate linear influence models to capture the risk dependencies between neighboring nodes, which have been used extensively to study risk propagation over networks \cite{alpcan2010network, miura2008security, nguyen2009stochastic}. The periodic strategy and linear influence model enable us to solve the lower-level security games and further analyze the higher-level insurance problems. The obtained results yield critical insights on network topology and insurance contract designation, and they provide valuable baselines for future analysis on both \texttt{FlipIt} games and cyber insurance.} 

\begin{figure}[]
\centering
{\includegraphics[width=0.49\textwidth]{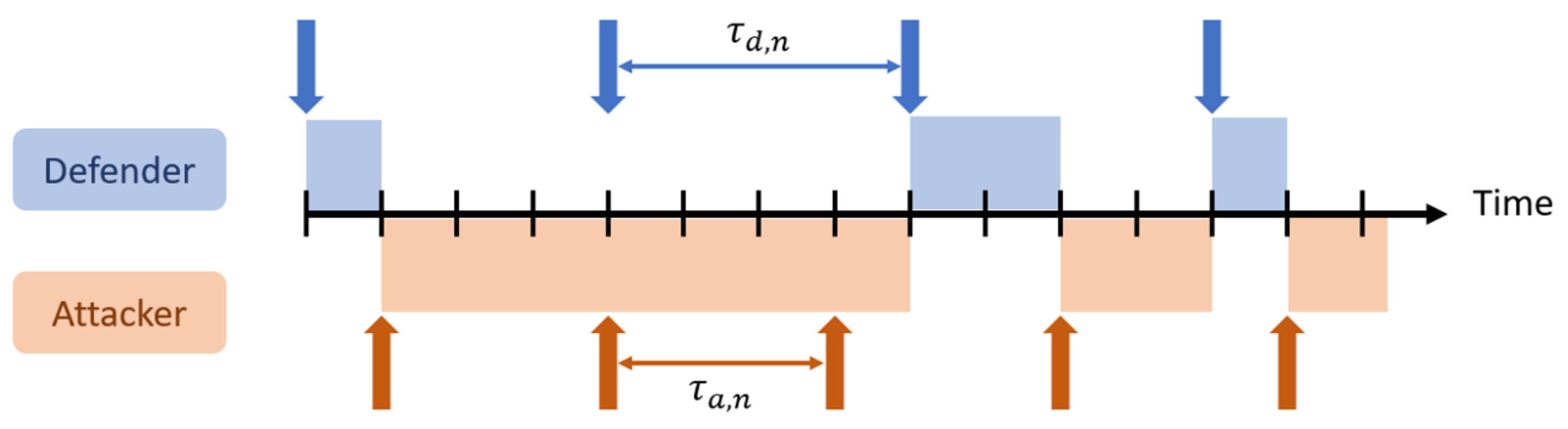}}
\caption{{The \texttt{FlipIt} game on Node $n$ between a defender and an attacker with both players adopting non-adaptive periodic strategies.}}
\label{fig:FlipItp}
\end{figure}

{Both the defender and the attacker adopt non-adaptive periodic strategies, i.e., both players have fixed intervals $\tau_{d,n} \in \mathbb{R}_{> 0}$ and $\tau_{a, n} \in \mathbb{R}_{> 0}$ between two consecutive moves as shown in Fig. \ref{fig:FlipItp}, respectively. In the following sections, we abuse the notations of strategies $p_{d, n}= \frac{1}{\tau_{d,n}}$ and $p_{a, n}=\frac{1}{\tau_{a,n}}$ to denote the defending frequency and the attacking frequency, respectively. We can compute $\alpha_n$, i.e., the expected proportion of the time that the attacker controls node $n$ by following the arguments in Section 4.1 in \cite{van2013flipit}. When $p_{d,n} \geq p_{a, n}$, i.e., $\tau_{d,n} \leq \tau_{a,n}$, the probability that the attacker moves in a given defender's move interval $\tau_{d,n}'$ is $p_{a,n}/p_{d,n}$; moreover, he moves exactly once within $\tau_{d,n}'$ since $\tau_{d, n} \leq \tau_{a, n}$ and his move is uniformly distributed at random within $\tau_{d,n}'$. Thus, we obtain $\alpha_n = \frac{p_{a,n}}{2p_{d,n}}$. Similarly, when $p_{d, n} < p_{a, n}$, we obtain $\alpha_n = 1-\frac{p_{d,n}}{2p_{a,n}}$. As a result, we have
\begin{equation}
\label{eq:Period}
\alpha_n = \left\lbrace  \begin{array}{c}
\begin{array}{lc}
{0,}&{p_{a,n} = 0;} 
\\ {\frac{p_{a, n}}{2p_{d, n}},}&{p_{d, n} \geq p_{a, n} > 0;}
\\ {1-\frac{p_{d, n}}{2p_{a, n}},}&{p_{a, n} > p_{d, n} \geq 0.}
\end{array}
\end{array}  \right.
\end{equation}
Note that when $p_{a, n} = 0$, i.e., there is no attacker or the attacker chooses not to attack, we have $\alpha_{n} = 0$ for $p_{d, n} \geq 0$ which indicates that the IoT device is always controlled by the defender. We can see from (\ref{eq:Period}) that the attacker has a larger $\alpha_n$ with the increase of his frequency $p_{u, n}$ and the decrease of the defender's frequency $p_{d, n}$. We could also see that $\alpha_n$ is continuous in both $p_{a,n}$ and $p_{d,n}$ as $\frac{p_{a, n}}{2p_{d, n}} = 1-\frac{p_{d, n}}{2p_{a, n}} = \frac{1}{2}$ when $p_{a,n} = p_{d, n}$. }

{We capture the risk dependencies between neighboring nodes with linear influence models. We use the following remark to illustrate linear influence models. 
\begin{remark}[Linear Influence Models]
\label{rem:LinearInfluenceModel}
The defender's risk level at node $n$ can be expressed as
 \begin{equation}
\label{eq:DefenderR}
R_n = \alpha_n + \eta \sum\limits_{m=1}^N w_{mn} R_m.
\end{equation} 
The first term is the expected proportion of the attacker's controlling time at node $n$, and a higher proportion indicates a higher risk level of the defender at this node. The second term captures the risks caused by neighboring nodes.  The parameter $\eta \in [0, 1]$ denotes the discount ratio of the network influence, and a larger $\eta$ denotes a stronger influence from neighboring nodes and indicates that the network is strongly connected; the parameters $w_{mn}\in [0, 1]$ denote the probability that node $n$ is attacked by the attacker at its neighboring node $m$ and they satisfy
\begin{equation}
\label{eq:w}
w_{nn} = 0,\  \sum\limits_{n=1}^N w_{mn} = 1.
\end{equation} 
Note that we have $w_{nn} = 0$ as the influence of the attacker at node $n$ has been captured by $\alpha_n$. 
\end{remark}
With linear influence models, we could achieve the following remark from Proposition 6 in \cite{zhang2017bi}.
\begin{remark}
\label{rem:DefenderRSolve}
Let $\mathbf{W}$ denote the network matrix with the $m$-th row and $n$-th column being $w_{mn}$ and let $\mathbf{I}_N$ denote an identity matrix of size $N$, we have
 \begin{equation}
\label{eq:DefenderRSolve}
R_n = g_n(\alpha_1, \alpha_2, ..., \alpha_N) = \sum\limits_{m=1}^N w_{nm}^* \alpha_m,
\end{equation} 
where $w_{nm}^*$ is the element at the $n$-th row and $m$-th column of matrix $\mathbf{W}^* = (\mathbf{I}_N-\eta \mathbf{W}^T)^{-1}$ with scalar $\eta \in [0, 1]$ being the discount ratio of the network influence. The matrix $\mathbf{W}^*$ is valid as the inverse of $\mathbf{I}_N-\eta \mathbf{W}^T$ exists. Furthermore, $w_{nm}^*$ satisfies
\begin{itemize}
\item[(i)] $w_{nn}^* > 1$ and $w_{nm}^* \geq 0$ for all $n, m \in \mathcal{N}$;
\item[(ii)] $\sum\limits_{m=1}^N w_{mn}^* = \frac{1}{1-\eta}$ for all $n \in \mathcal{N}$. 
\end{itemize}
\end{remark}
Remark \ref{rem:DefenderRSolve} indicates that the defender's risk level at one node is also affected by $\alpha_m$ in other nodes. The defender's risk level is higher when any attacker in this network has a larger expected proportion of controlling time, which captures a negative impact of the network influence on cyber security. }

{
\begin{remark}
\label{rem:DefenderLoss}
From (\ref{eq:Period}) and Remark \ref{rem:DefenderRSolve}, the defender at node $n$ has a higher risk level or faces a larger expected loss if
\begin{itemize}
\item the attacker at node $n$ has a larger expected proportion of controlling time, i.e., $\alpha_n$ is larger;
\item the attacker at node $n$ attacks more frequently, i.e., the attacker has a larger $p_{a, n}$;
\item the defender at node $n$ defends less frequently, i.e, the defender has a smaller $p_{d, n}$;
\item the attackers at other nodes have larger expected proportions of controlling time, i.e., $\alpha_m$ is larger for $m\neq n$. 
\end{itemize}
\end{remark}}

The players' problems in both Defender-D and Defender-C under periodic strategies and linear influence models could be obtained by plugging (\ref{eq:Period}) and (\ref{eq:DefenderRSolve}) into the corresponding problems.  Since solving the insurer's problem relies on the results of the \texttt{FlipIt} games, we first solve the lower-level \texttt{FlipIt} games and obtain the reactions of both defenders and attackers to the insurance contracts. Then, we solve the insurer's problem and obtain optimal insurance contracts. Note that we consider linear costs of defenders and attackers in the following sections, and we abuse the notations of $c_{d, n}$ and $c_{a,n}$ to denote the corresponding cost parameters. The attacker's problem (\ref{eq:AttackernMax}) can now be written as 
\begin{equation}
\label{eq:AttackernMaxp}
\max\limits_{p_{a, n}} \  \gamma_{a, n} \alpha_n - c_{a, n}p_{a, n}.
\end{equation}
Different $\gamma_{a, n}$ and $c_{a, n}\in\mathbb{R}_{\geq 0}$ capture the trade-offs between a larger proportion of time and a smaller attacking frequency of the attacker. 

{We could obtain the following problem after plugging (\ref{eq:DefenderRSolve}) into the defender's problem (\ref{eq:DefendernMin}) in Defender-D. 
\begin{equation}
\label{eq:DefendernMinp}
\min\limits_{p_{d, n}} \  (1-s_n)\gamma_{d, n}\sum\limits_{m=1}^N w_{nm}^*\alpha_m + c_{d, n}p_{d, n}. 
\end{equation}
The parameters $\gamma_{a, n}$ and $c_{d, n}\in\mathbb{R}_{\geq 0}$ capture the trade-offs between a smaller expected effective loss and a larger defending frequency of the defender.
\begin{remark}[Distributed Computations]
\label{rem:DefendernMin}
Since $\alpha_{m}$ is a constant with respect to $p_{d, n}$ if $m \neq n$, problem (\ref{eq:DefendernMinp}) can be simplified further into the following problem
\begin{equation}
\label{eq:DefendernMinSimplified}
\min\limits_{p_{d, n}} \  (1-s_n)\gamma_{d, n} w_{nn}^* \alpha_n + c_{d, n}p_{d, n} .
\end{equation} 
Problem (\ref{eq:DefendernMinSimplified}) indicates that the defender's decision on $p_{d, n}$ in one \texttt{FlipIt-D} game is not affected by results of other \texttt{FlipIt-D} games given the coverage level $s_n$. However, the expected loss of the defender $\mathbb{E}[X_n]$ is still affected by the outcomes of other \texttt{FlipIt-D} games.  
\end{remark}

{We could obtain the following problem after plugging (\ref{eq:DefenderRSolve}) into the defender's problem (\ref{eq:DefenderNMin}) in Defender-C. 
\begin{equation}
\label{eq:DefenderNMinp}
\min\limits_{\{p_{d, n}\}}\ (1-s)\sum\limits_{n=1}^N \sum\limits_{m=1}^N \gamma_{d,n}w_{nm}^*\alpha_m+ \sum\limits_{n=1}^Nc_{d, n}p_{d, n}.
\end{equation}
Note that $g_n(\alpha_1, ..., \alpha_N)$ in (\ref{eq:DefenderRSolve}) is additively separable for all $1 \leq n \leq N$, thus, the \texttt{L-FlipIt-C} games exist under periodic strategy and linear influence model and problem (\ref{eq:DefenderNMinp}) is equivalent to the following $N$ sub-problems from Remark \ref{rem:DefenderNMin}.
\begin{equation}
\label{eq:DefenderNMinSubp}
\min\limits_{p_{d, n}}\ (1-s)\sum_{m=1}^N \gamma_{d,m} w_{mn}^*\alpha_n +c_{d, n}p_{d, n} .
\end{equation}}

{Note that the \texttt{FlipIt-D} games and the \texttt{L-FlipIt-C} games share similar structures: the attackers in both cases solve the same optimization problems (\ref{eq:AttackernMaxp}); the defenders' problems (\ref{eq:DefendernMinSimplified}) and (\ref{eq:DefenderNMinSubp}) can be written into one unified optimization problem as
\begin{equation}
\label{eq:DefenderMin}
\min\limits_{p_{d, n}} \  (1- \tilde{s}_n) \tilde{\gamma}_{d,n} \alpha_n + c_{d, n} p_{d, n},
\end{equation}
where
\begin{equation}
\label{eq:Defenders}
\tilde{s}_{n}=\left\lbrace \begin{array}{lc}
{s_n,}&{\text{Defender-D};}\\
{s,}&{\text{Defender-C},}
\end{array}  \right.
\end{equation}
\begin{equation}
\label{eq:DefenderGamma}
\tilde{\gamma}_{d,n}=\left\lbrace \begin{array}{lc}
{\gamma_{d, n} w_{nn}^*,}&{\text{Defender-D};}\\
{\sum\limits_{m=1}^N \gamma_{d, m} w_{mn}^*,}&{\text{Defender-C}.}
\end{array}  \right.
\end{equation}
\begin{remark}
\label{rem:DefendernN}
Problem (\ref{eq:DefenderMin}) can be interpreted that the defender aims to minimize the expected proportion of the attacker's controlling time. Thus, given the same coverage level on an IoT device, the defender in Defender-C cares more about reducing the impacts from the attackers compared with the defenders in Defender-D as $\sum_{m=1}^N \gamma_{d, m} w_{mn}^* \geq \gamma_{d, n} w_{nn}^*$. 
\end{remark}
We can further define an unified local \texttt{FlipIt} game as follows.
\begin{definition}
\label{def:FlipItLocal}
Let $\mathcal{S}_{a, n} = \{p_{a,n}| p_{a,n}\in\mathbb{R}_{\geq 0}\}$ and $\mathcal{S}_{d, n} = \{p_{d,n}| p_{d,n}\in\mathbb{R}_{\geq 0}\}$ denote the action sets for the attacker and the defender at node $n$, respectively; let $J_{a,n}(p_{a,n},p_{d,n})$ and $J_{d,n}(p_{d,n},p_{a,n}; \tilde{s}_n)$ denote the objective functions from (\ref{eq:AttackernMax}) and (\ref{eq:DefenderMin}), respectively. 
A strategy profile $\{p_{a,n}^*, p_{d,n}^*\}$ is a Nash equilibrium of the local \texttt{FlipIt} game at node $n$ defined by $\left\langle\{\text{Attacker, Defender}\},\{\mathcal{S}_{a,n},\mathcal{S}_{d,n}\}, \{J_{a,n}, J_{d,n}\}\right\rangle$ if 
\[ J_{a, n}(p_{a,n}^*, p_{d,n}^*) \geq J_{a, n}(p_{a,n}, p_{d,n}^*), \ \  \forall p_{a,n} \in \mathcal{S}_{a,n};\]
\[J_{d, n}(p_{d,n}^*, p_{a,n}^*;\tilde{s}_n)   \leq J_{d, n}(p_{d,n}, p_{a,n}^*;\tilde{s}_n), \ \  \forall p_{d,n} \in \mathcal{S}_{d,n}. \]
\end{definition}
We can obtain the solutions of the \texttt{FlipIt-D} games and the \texttt{L-FlipIt-C} games by plugging (\ref{eq:Defenders}) and (\ref{eq:DefenderGamma}) into the solution of the local \texttt{FlipIt} game defined in Definition \ref{def:FlipItLocal}. The solutions of the \texttt{G-FlipIt-D} game and the \texttt{FlipIt-C} game can be further obtained by their Definitions \ref{def:FlipItn} and \ref{def:FlipItN}, respectively. Thus, we can solve all \texttt{FlipIt} games in both Defender-D and Defender-C by solving the local \texttt{FlipIt} game. }

\begin{figure}[]
\centering
\subfigure[$\frac{\gamma_{a, n}}{2c_{a, n}} \geq \frac{(1-\tilde{s}_n)\tilde{\gamma}_{d,n}}{2c_{d, n}}$]{\includegraphics[width=0.24\textwidth]{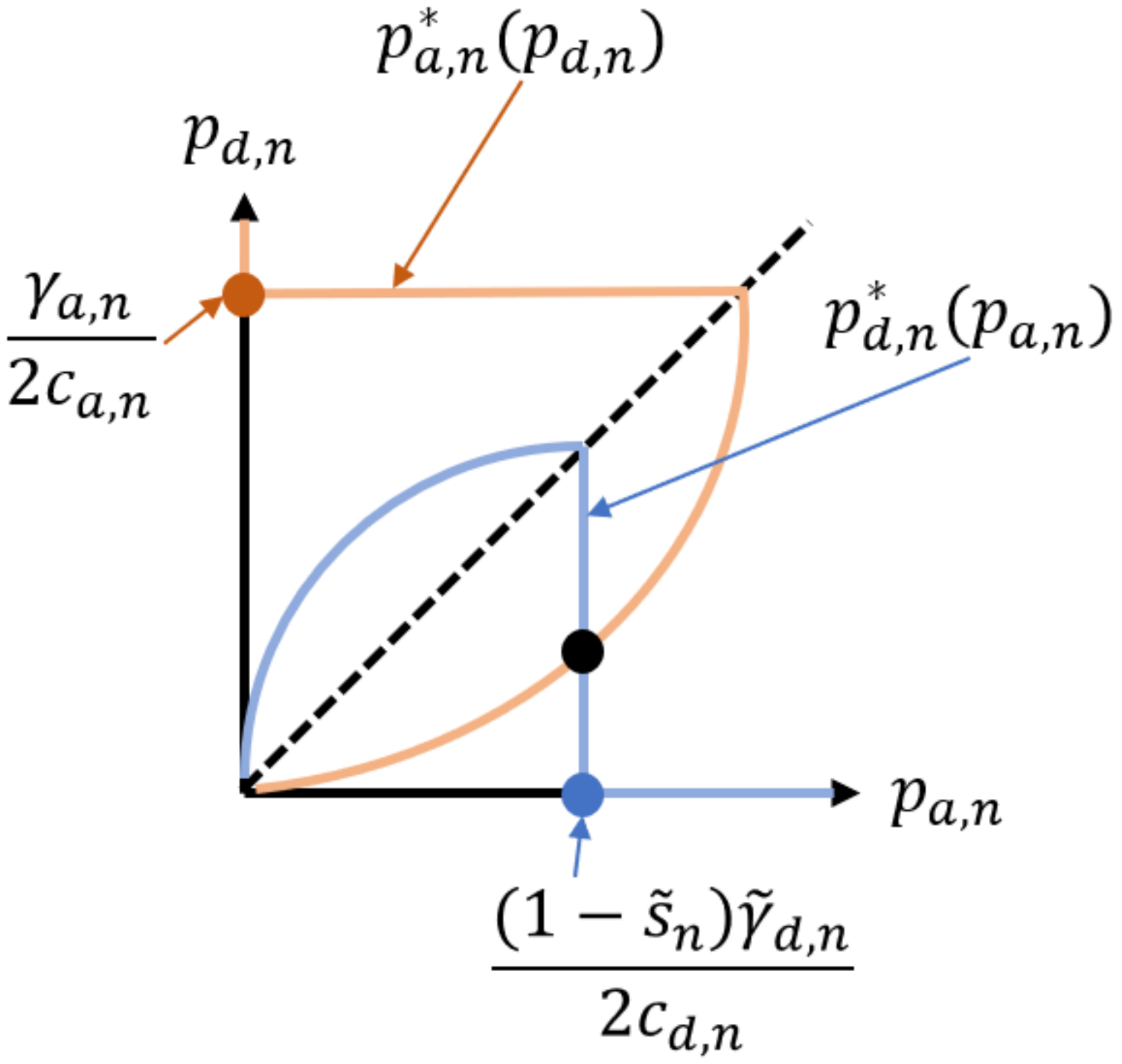}}
\subfigure[$\frac{\gamma_{a, n}}{2c_{a, n}} < \frac{(1-\tilde{s}_n)\tilde{\gamma}_{d,n}}{2c_{d, n}}$]{\includegraphics[width=0.24\textwidth]{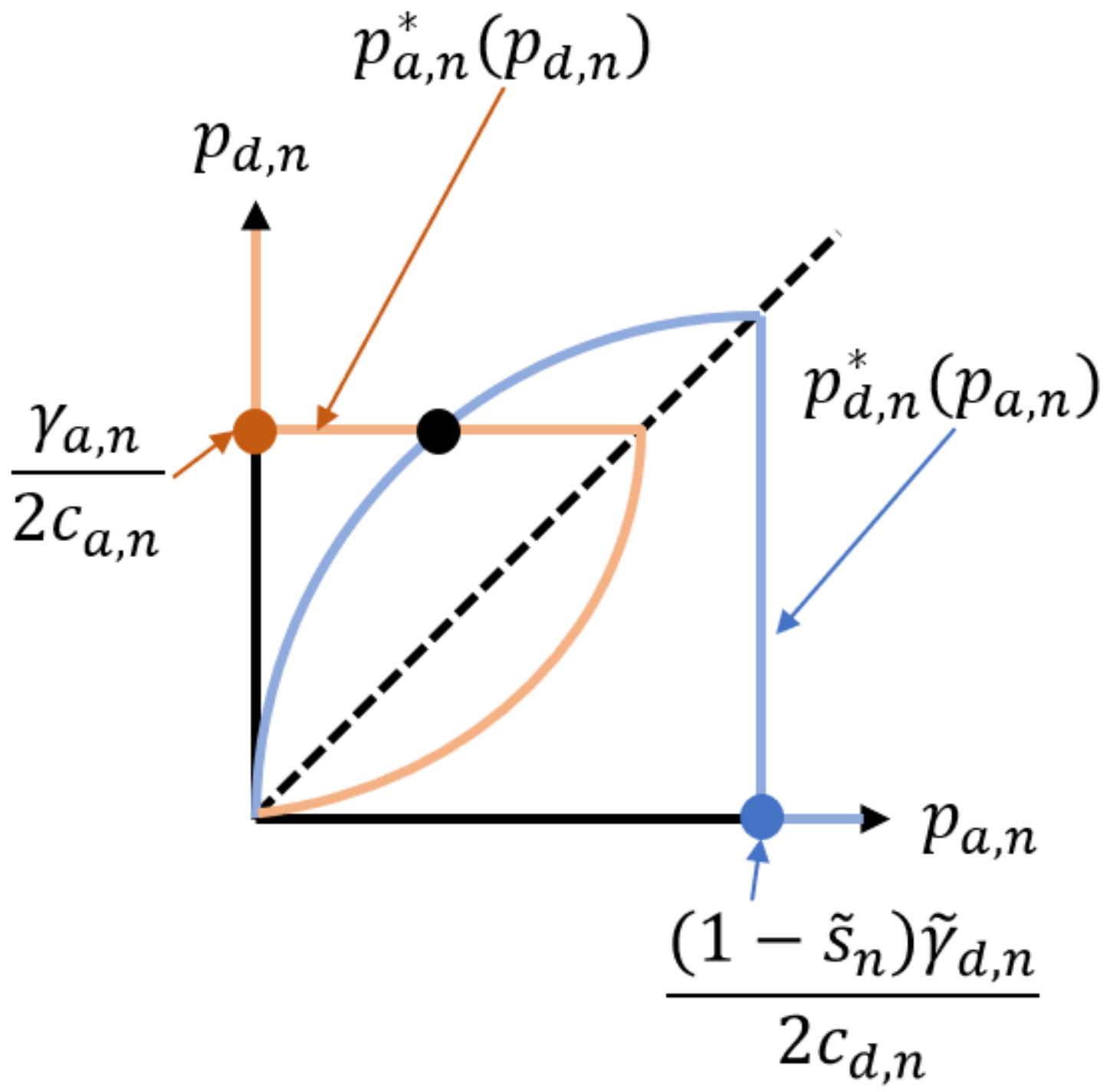}}
\caption{Equilibrium of the local \texttt{FlipIt} game at node $n$.  $p_{a,n}^*(p_{d,n})$ denotes the best response of the attacker given the defender's defending frequency $p_{d,n}$ while $p_{d,n}^*(p_{a,n})$ denotes the best response of the defender given the attacker's attacking frequency $p_{a,n}$. }
\label{fig:Equi}
\end{figure}

{The local \texttt{FlipIt} game defined in Definition \ref{def:FlipItLocal} is different from the original \texttt{FlipIt} game presented in \cite{van2013flipit} as the defender here aims to minimize the losses caused by the attacker while the defender in the original \texttt{FlipIt} game aims to maximize the proportion of his controlling time. Following similar steps as in \cite{van2013flipit}, we can obtain the equilibrium of the local \texttt{FlipIt} game defined in Definition \ref{def:FlipItLocal} by finding the intersection of the players' best responses as shown in Fig. \ref{fig:Equi}.
\begin{proposition}
\label{pro:FlipItEqu}
The Nash equilibrium of the local \texttt{FlipIt} game defined in Definition \ref{def:FlipItLocal} can be summarized into two different cases as shown in Fig. \ref{fig:Equi}.
\begin{itemize}
\item If $\frac{\gamma_{a, n}}{2c_{a, n}} \geq \frac{(1-\tilde{s}_n)\tilde{\gamma}_{d,n}}{2c_{d, n}}$, the equilibrium (\textbf{FlipIt-E1}) is achieved at
\begin{equation}
\label{eq:FlipItEqu>}
p_{d, n}^* = \frac{(1-\tilde{s}_n)^2\tilde{\gamma}_{d,n}^2c_{a, n}}{2\gamma_{a, n}c_{d, n}^2}, \ \ p_{a, n}^* = \frac{(1-\tilde{s}_n)\tilde{\gamma}_{d,n}}{2c_{d, n}};
\end{equation}
\item If $\frac{\gamma_{a, n}}{2c_{a, n}} < \frac{(1-\tilde{s}_n)\tilde{\gamma}_{d,n}}{2c_{d, n}}$, the equilibrium (\textbf{FlipIt-E2}) is achieved at 
\begin{equation}
\label{eq:FlipItEqu<}
p_{d, n}^* = \frac{\gamma_{a, n}}{2c_{a, n}}, \ \ p_{a, n}^* = \frac{\gamma_{a, n}^2c_{d, n}}{2(1-\tilde{s}_n)\tilde{\gamma}_{d,n} c_{a, n}^2}.
\end{equation}
\end{itemize}
\end{proposition}
\begin{proof}
The equilibrium can be obtained by finding the intersection between the best responses as shown in Fig. \ref{fig:Equi}. 
\end{proof}
Note that $p_{d,n}^*= 0$ and $p_{a,n}^* = 0$ are also the intersection of the best responses. However, we exclude them in this paper as there are no defender and attacker when $p_{d,n}^*= 0$ and $p_{a,n}^* = 0$. We can see from Proposition \ref{pro:FlipItEqu} that the equilibrium is affected by the coverage level $\tilde{s}_n$, and we have the following remarks regarding the relations between them. }
\begin{remark}[Equilibrium Shift]
\label{rem:Shift}
{If $\gamma_{a, n}c_{d, n} \geq \tilde{\gamma}_{d,n}c_{a, n}$, we have $\frac{\gamma_{a, n}}{2c_{a, n}} \geq \frac{(1-\tilde{s}_n)\tilde{\gamma}_{d,n}}{2c_{d, n}}$ for $0< \tilde{s}_n \leq 1$, and thus the equilibrium is \textit{FlipIt-E1} for $0< \tilde{s}_n \leq 1$.}

{If $\gamma_{a, n}c_{d, n} < \tilde{\gamma}_{d,n}c_{a, n}$, we note that $\frac{\gamma_{a, n}}{2c_{a, n}} < \frac{(1-\tilde{s}_n)\tilde{\gamma}_{d,n}}{2c_{d, n}}$ when $0< \tilde{s}_n<1-\frac{\gamma_{a, n}c_{d, n}}{\tilde{\gamma}_{d,n}c_{a, n}}$, and thus the equilibrium is \textit{FlipIt-E2}. However, we have $\frac{\gamma_{a, n}}{2c_{a, n}} \geq \frac{(1-\tilde{s}_n)\tilde{\gamma}_{d,n}}{2c_{d, n}}$ when $1-\frac{\gamma_{a, n}c_{d, n}}{\tilde{\gamma}_{d,n}c_{a, n}} \leq \tilde{s_n} \leq 1$, and thus the equilibrium is \textit{FlipIt-E1}. As a result, the equilibrium shifts from \textit{FlipIt-E2} to \textit{FlipIt-E1} as the coverage level increases. }
\end{remark}
\begin{remark}[Risk Compensation and Peltzman Effect]
\label{rem:Peltzman}
{At \textit{FlipIt-E1}, both the defender and the attacker reduce their frequencies as the coverage level increases. The defender's reckless behavior under the insurance in this case is referred as risk compensation \cite{ewold1991insurance}. The proportion of the attacker's controlling time $\alpha_n^*=1-\frac{(1-\tilde{s}_n)\tilde{\gamma}_{d,n}c_{a, n}}{2\gamma_{a, n}c_{d, n}}$ increases with the coverage level, as a result, the defender faces a higher risk, and such phenomena under the insurance is referred as Peltzman effect \cite{peltzman1975effects}. }

{However, at \textit{FlipIt-E2}, the defender does not change his frequency while the attacker increases his frequency as the coverage level increases. The proportion of the attacker's controlling time $\alpha_n^*=\frac{\gamma_{a, n}c_{d, n}}{2(1-\tilde{s}_n)\tilde{\gamma}_{d,n}c_{a, n}}$ increases with the coverage level. Thus, at \textit{FlipIt-E2}, there is no risk compensation, but we can observe Peltzman effect. }
\end{remark}
{
With the results of the local \texttt{FlipIt} game, we can solve the insurer's problems and obtain the optimal insurance contracts. In the following sections, we discuss cyber insurance separately for Defender-D and Defender-C. }

\section{Cyber Insurance: Defender-D}
\label{sec:D}
In this section, we aim to solve the insurer's problem (\ref{eq:Insurern}) in Defender-D. We can obtain the results of the \texttt{FlipIt-D} game by plugging $\tilde{s}_n = s_n$ and $\tilde{\gamma}_{d,n}=\gamma_{d,n}w_{nn}^*$ into Proposition \ref{pro:FlipItEqu}. Let us abuse the notations $p_{d,n}^*(s_n)$, $p_{a,n}^*(s_n)$, and $\alpha_{n}^*(s_n)$ to denote the equilibrium results under the coverage level $s_n$; let $K_{d,n}^*(s_n) = J_{d,n}(p_{d,n}^*,p_{a,n}^*;s_n)$ from Definition \ref{def:FlipItLocal}. Note that $s_n=0$ indicates the results under no insurance. 
{
\begin{remark}
\label{rem:DefendernIR}
The defender's decision on $p_{d, n}$ is not affected by the players' decisions at other nodes from Remark \ref{rem:DefendernMin} while the attacker's decision on $p_{a,n}$ is also not affected by the players' decisions at other nodes from (\ref{eq:AttackernMax}). Thus, we have $\alpha_m' = \alpha_m, \forall m \neq n$ in (\ref{eq:DefendernIR}). As a result, (\ref{eq:DefendernIR}) can be rewritten as
\begin{equation}
\label{eq:DefenderIRn}
\begin{array}{l}
T_n \leq K_{d,n}^*(0)-K_{d,n}^*(s_n)+s_n\gamma_{d, n}\sum\limits_{m\neq n} w_{nm}^* \alpha_m.
\end{array}
\end{equation}
\end{remark}}

Since the insurer aims to maximize his profit, he sets highest possible premium at
\begin{equation}
\label{eq:InsurerMaxn}
T_{n,\max} = K_{d,n}^*(0)-K_{d,n}^*(s_n)+s_n\gamma_{d, n}\sum\limits_{m\neq n} w_{nm}^* \alpha_m.
\end{equation}
As a result, solving the insurer's problem (\ref{eq:Insurern}) is equivalent to solving the following problem after plugging (\ref{eq:InsurerMaxn}) into (\ref{eq:Insurern}): 
\begin{equation}
\label{eq:Insurancen}
\begin{array}{l}
s_n^* = \arg\max\limits_{0< s_n\leq 1}\  K_{d,n}^*(0)-K_{d,n}^*(s_n)-s_n\gamma_{d, n}w_{nn}^* \alpha_n^*(s_n)\\
\begin{array}{cc}
{\text{s.t.}}&{K_{d,n}^*(0)-K_{d,n}^*(s_n) - s_n\gamma_{d, n}w_{nn}^* \alpha_n^*(s_n) \geq 0,}
\end{array}
\end{array}
\end{equation}
and the premium $T_n^*$ can be achieved by plugging $s_n^*$ into (\ref{eq:InsurerMaxn}). Since we have achieved the equilibrium results of the \texttt{FlipIt-D} game in the previous section, we can directly solve (\ref{eq:Insurancen}). 

\subsection{High-Risk Regime: $\gamma_{a, n}c_{d, n} \geq \gamma_{d, n}w_{nn}^*c_{a, n}$}
In this case, the \texttt{FlipIt-D} game between the defender and the attacker achieves \textit{FlipIt-E1} as in Remark \ref{rem:Shift}. After plugging the results of \textit{FlipIt-E1}, problem (\ref{eq:Insurancen}) can be expressed as
\begin{equation}
\label{eq:Insurancen>}
\begin{array}{l}
s_n^*=\arg\max\limits_{0< s_n\leq 1} \  \frac{\gamma_{d, n}^2 w_{nn}^{*2} c_{a, n}}{2\gamma_{a, n} c_{d, n}}(1-s_n)s_n \\
\begin{array}{cc}
{\text{s.t.}}&{\frac{ \gamma_{d, n}^2 w_{nn}^{*2} c_{a, n}}{2\gamma_{a, n} c_{d, n}} (1-s_n)s_n \geq 0. }
\end{array}
\end{array}
\end{equation}
Note that $\frac{\gamma_{d, n}^2 w_{nn}^{*2} c_{a, n}}{2\gamma_{a, n} c_{d, n}}$ is constant for $s_n$ and the constraint is satisfied for $0< s_n \leq 1$. Thus, we only need to find $s_n^*$ that minimizes the objective function to obtain the optimal insurance contract.  
\begin{lemma}
\label{lem:InsurerOptimaln>}
If $\gamma_{a, n}c_{d, n} \geq \gamma_{d, n}w_{nn}^*c_{a, n}$, the optimal insurance contract is
\begin{equation}
\label{eq:InsurerOptimaln>}
s_n^* = \frac{1}{2}, \ \ \ T_n^* =  \frac{\gamma_{d,n}w_{nn}^*+\gamma_{d, n}\sum_{m\neq n} w_{nm}^* \alpha_m}{2}.
\end{equation}
The insurer's profit under this contract is $\frac{ \gamma_{d, n}^2 w_{nn}^{*2} c_{a, n}}{8\gamma_{a, n} c_{d, n}}$. 
\end{lemma}
\begin{proof}
We can achieve that $s_n^* = \frac{1}{2}$. $T_n^*$ can be achieved by plugging $s_n^*$ into (\ref{eq:InsurerMaxn}). 
\end{proof}

\subsection{Low-Risk Regime: $\gamma_{a, n}c_{d, n} < \gamma_{d, n}w_{nn}^*c_{a, n}$}
In this case, the equilibrium of the \texttt{FlipIt-D} game shifts from \textit{FlipIt-E2} to \textit{FlipIt-E1} as the coverage level increases from Remark \ref{rem:Shift}. When $0< s_n < 1-\frac{\gamma_{a, n}c_{d, n}}{\gamma_{d, n}w_{nn}^*c_{a, n}}$, the \texttt{FlipIt-D} game achieves \textit{FlipIt-E2} and we have $T_{n,\max}=K_{d,n}^*(0)-K_{d,n}^*(s_n) - s_n\gamma_{d, n}w_{nn}^* \alpha_n^*(s_n) = - s_n\gamma_{d, n}w_{nn}^* \alpha_n^*(s_n) < 0$. Thus, the insurer does not provide any insurance contracts with $0< s_n < 1-\frac{\gamma_{a, n}c_{d, n}}{\gamma_{d, n}w_{nn}^*c_{a, n}}$.

When $1-\frac{\gamma_{a, n}c_{d, n}}{\gamma_{d, n}w_{nn}^*c_{a, n}} \leq s_n \leq 1$, the \texttt{FlipIt} game achieves \textit{FlipIt-E1} under the insurance and \textit{FlipIt-E2} without the insurance, and after plugging the results of \textit{FlipIt-E1} and \textit{FlipIt-E2} into (\ref{eq:Insurancen}), we have
\begin{equation}
\label{eq:Insurancen<}
\begin{array}{l}
s_n^* \in \arg\max\limits_{s_n}\ \frac{\gamma_{a, n}c_{d, n}}{c_{a, n}}-\gamma_{d, n} w_{nn}^* + \frac{ \gamma_{d, n} ^2w_{nn}^{*2} c_{a, n}}{2\gamma_{a, n} c_{d, n}}(1-s_n)s_n \\
\begin{array}{cc}
{\text{s.t.}}&{\begin{array}{c}
\frac{\gamma_{a, n}c_{d, n}}{c_{a, n}}-\gamma_{d, n} w_{nn}^* + \frac{ \gamma_{d, n} ^2w_{nn}^{*2} c_{a, n}}{2\gamma_{a, n} c_{d, n}}(1-s_n)s_n  \geq 0.
\end{array}}
\end{array}
\end{array}
\end{equation}
We first obtain the following proposition regarding the insurability of the defender. 
\begin{proposition}[Insurability]
\label{pro:Insurability}
The defender is not insurable, i.e., there exists no effective insurance contract and the equilibrium of the \texttt{FlipIn-D} game does not exist, if 
\begin{equation}
\label{eq:Insurability}
0 < \frac{\gamma_{a, n}c_{d, n}}{\gamma_{d, n}w_{nn}^*c_{a, n}} < \frac{1}{2} + \frac{\sqrt{2}}{4}.
\end{equation}
\end{proposition}
\begin{proof}
See Appendix A. 
\end{proof}
Proposition \ref{pro:Insurability} comes from the individual rationality constraints of both insurer and defender, and it reflects situations that the insurer has no incentive to provide insurance to the defender as he cannot make a profit from it or the defender has no incentive to accept any insurance as he has larger costs with it. We can further achieve the following remark regarding the condition (\ref{eq:Insurability}).
\begin{remark}
	\label{rem:InsurableEqu}
	The defender is not insurable if:
	\begin{itemize}
		\item[(i)] $\gamma_{d,n}$ is high, i.e., the attacker inflicts large losses on the defender;
		\item[(ii)] $c_{d,n}$ is low, i.e., the defender has a low cost to control the device frequently; 
		\item[(iii)] $\gamma_{a,n}$ is low, i.e., the attacker has a low benefit of controlling the device;
		\item[(iv)] $c_{a,n}$ is high, i.e., the attacker has a high cost to control the device frequently;
		\item[(v)] $w_{nn}^*$ is high, i.e., the network effect is high.  
	\end{itemize}
\end{remark}
We have that the following proposition regarding the optimal insurance contracts when the defender is insurable.
\begin{lemma}
\label{lem:InsurerOptimaln<}
If $\frac{1}{2} + \frac{\sqrt{2}}{4}\leq  \frac{\gamma_{a, n}c_{d, n}}{\gamma_{d, n}w_{nn}^*c_{a, n}} <1$, the optimal insurance contract is
\begin{equation}
\label{eq:InsurerOptimaln<1}
s_n^* = \frac{1}{2}, \ \ \ T_n^* =  \frac{\gamma_{a, n}c_{d, n}}{c_{a, n}}-\frac{\gamma_{d, n} w_{nn}^*}{2} + \frac{\gamma_{d, n}\sum_{m\neq n} w_{nm}^* \alpha_m}{2}.
\end{equation}
The insurer's profit under this contract is $\frac{\gamma_{a, n}c_{d, n}}{c_{a, n}}-\gamma_{d, n} w_{nn}^* + \frac{ \gamma_{d, n} ^2w_{nn}^{*2} c_{a, n}}{8\gamma_{a, n} c_{d, n}} $. 
\end{lemma}
\begin{proof}
We can achieve that $s_n^* = \frac{1}{2}$ from (\ref{eq:Insurancen<}), and it satisfies the constraint $1-\frac{\gamma_{a, n}c_{d, n}}{\gamma_{d, n}w_{nn}^*c_{a, n}}  \leq s_n \leq 1$. $T_n^*$ can be obtained by plugging $s_n^*$ into (\ref{eq:InsurerMaxn}).
\end{proof}

With Lemma \ref{lem:InsurerOptimaln>}, Proposition \ref{pro:Insurability}, and Lemma \ref{lem:InsurerOptimaln<}, we have the following proposition regarding the equilibrium of the \texttt{FlipIn-D} game defined in Definition \ref{def:BiLeveln}. 
\begin{proposition}
\label{pro:EquBin}
The Nash equilibrium of the \texttt{FlipIn-D} game defined in Definition \ref{def:BiLeveln} can be summarized into the following three cases:
\begin{itemize}
\item if $\frac{\gamma_{a, n}c_{d, n}}{\gamma_{d, n}w_{nn}^*c_{a, n}}\geq 1$, the equilibrium is achieved at
\[s_n^* = \frac{1}{2}, T_n^* =  \frac{\gamma_{d,n}w_{nn}^*}{2}+\frac{\gamma_{d, n}\sum_{m\neq n} w_{nm}^* \alpha_m(s_m)}{2},\]
\[p_{d, n}^* = \frac{ \gamma_{d, n}^2 w_{nn}^{*2}c_{a, n}}{8\gamma_{a, n}c_{d, n}^2},p_{a, n}^* =  \frac{\gamma_{d, n} w_{nn}^*}{4c_{d, n}};\]
\item if $\frac{1}{2} + \frac{\sqrt{2}}{4}\leq \frac{\gamma_{a, n}c_{d, n}}{\gamma_{d, n}w_{nn}^*c_{a, n}} <1$, the equilibrium is achieved at 
\[s_n^* = \frac{1}{2},  T_n^* =  \frac{\gamma_{a, n}c_{d, n}}{c_{a, n}}-\frac{\gamma_{d, n} w_{nn}^*}{2} + \frac{\gamma_{d, n}\sum_{m\neq n} w_{nm}^* \alpha_m(s_m)}{2}, \]
\[p_{d, n}^* = \frac{ \gamma_{d, n}^2 w_{nn}^{*2}c_{a, n}}{8\gamma_{a, n}c_{d, n}^2},p_{a, n}^* =  \frac{\gamma_{d, n} w_{nn}^*}{4c_{d, n}};\]
\item if $ \frac{\gamma_{a, n}c_{d, n}}{\gamma_{d, n}w_{nn}^*c_{a, n}} < \frac{1}{2} + \frac{\sqrt{2}}{4}$, the equilibrium does not exist. The defender and the attacker have 
\[p_{d,n}^* = \frac{\gamma_{a,n}}{2c_{a,n}}, p_{a,n}^* = \frac{\gamma_{a, n}^2c_{d, n}}{2\gamma_{d,n}w_{nn}^* c_{a, n}^2}.\]
\end{itemize} 
\end{proposition}
\begin{proof}
This proposition follows from combining Proposition \ref{pro:FlipItEqu}, Lemma \ref{lem:InsurerOptimaln>}, Proposition \ref{pro:Insurability}, and Lemma \ref{lem:InsurerOptimaln<}.
\end{proof}
We can see that when the defender is insurable, the optimal insurance contract provides a coverage level of $\frac{1}{2}$.  
\begin{remark}
\label{remLEquiBinG}
The Nash equilibrium of the \texttt{G-FlipIn-D} game defined in Definition \ref{def:BiLeveln} could be obtained with Proposition \ref{pro:EquBin} by combing the results of all the \texttt{FlipIn-D} games.
\end{remark}
We could also obtain the Nash equilibrium of the \texttt{FlipIn-D} game when there are no network connectivities by following the similar steps in this section. In this case, all the IoT devices are not connected with each other or there is only one IoT device in this network. 
\begin{corollary}
\label{cor:EquBinNoNetwork}
When the network is not connected, the Nash equilibrium of the \texttt{FlipIn-D} game can be summarized into the following three cases:
\begin{itemize}
\item if $\frac{\gamma_{a, n}c_{d, n}}{\gamma_{d, n}c_{a, n}}\geq 1$, the equilibrium is achieved at $s_n^* = \frac{1}{2}, T_n^* =  \frac{\gamma_{d,n}}{2},p_{d, n}^* = \frac{ \gamma_{d, n}^2 c_{a, n}}{8\gamma_{a, n}c_{d, n}^2},p_{a, n}^* =  \frac{\gamma_{d, n} }{4c_{d, n}}$;
\item if $\frac{1}{2} + \frac{\sqrt{2}}{4}\leq \frac{\gamma_{a, n}c_{d, n}}{\gamma_{d, n}c_{a, n}}<1$, the equilibrium is achieved at $s_n^* = \frac{1}{2},  T_n^* =  \frac{\gamma_{a, n}c_{d, n}}{c_{a, n}}-\frac{\gamma_{d, n} }{2}, p_{d, n}^* = \frac{ \gamma_{d, n}^2 c_{a, n}}{8\gamma_{a, n}c_{d, n}^2},p_{a, n}^* =  \frac{\gamma_{d, n} }{4c_{d, n}}$;
\item if $ \frac{\gamma_{a, n}c_{d, n}}{\gamma_{d, n}c_{a, n}} < \frac{1}{2} + \frac{\sqrt{2}}{4}$, the equilibrium does not exist. The defender and the attacker have $p_{d,n}^* = \frac{\gamma_{a,n}}{2c_{a,n}}, p_{a,n}^* = \frac{\gamma_{a, n}^2c_{d, n}}{2\gamma_{d,n}c_{a, n}^2}$.
\end{itemize} 
\end{corollary}
Recall $w_{nn}^* > 1$ and $w_{nm}^* \geq 0$ from Remark \ref{rem:DefenderRSolve}. By comparing Proposition \ref{pro:EquBin} and Corollary \ref{cor:EquBinNoNetwork}, we can see that the network effect decreases the insurability as the insurable defender could be uninsurable because of network effects when $\frac{\gamma_{a, n}c_{d, n}}{\gamma_{d, n}w_{nn}^*c_{a, n}} < \frac{1}{2} + \frac{\sqrt{2}}{4} \leq \frac{\gamma_{a, n}c_{d, n}}{\gamma_{d, n}c_{a, n}} $. Moreover, the premium of the optimal insurance contract is also higher with network connectivity when $\frac{\gamma_{a, n}c_{d, n}}{\gamma_{d, n}w_{nn}^*c_{a, n}} \geq 1$.   

\section{Cyber Insurance: Defender-C}
\label{sec:C}
In this section, we analyze the insurer's problem (\ref{eq:InsurerN}) for the \texttt{FlipIn-C} game. Following similar steps in the previous section, let $K_{d,n}^*(s) = J_{d,n}(p_{d,n}^*,p_{a,n}^*;s)$ from Definition \ref{def:FlipItLocal}. The defender's individual rationality constraint (\ref{eq:DefenderNIR}) indicates that:
\begin{equation}
\label{eq:DefenderIRN}
\begin{array}{l}
T \leq \sum\limits_{n=1}^N \left(K_{d,n}^*(0)-K_{d,n}^*(s) \right).
\end{array}
\end{equation}
Thus, the highest premium that the insurer can charge is
\begin{equation}
\label{eq:DefenderMaxN}
\begin{array}{l}
T_{\max} =  \sum\limits_{n=1}^N \left(K_{d,n}^*(0)-K_{d,n}^*(s) \right).
\end{array}
\end{equation}
As a result, solving the insurer's problem (\ref{eq:InsurerN}) is equivalent to solving the following problem after plugging (\ref{eq:DefenderMaxN}) into (\ref{eq:InsurerN}):
\begin{equation}
\label{eq:InsuranceN}
\begin{array}{l}
s^*=\arg\max\limits_{0 < s \leq 1} \ \sum\limits_{n=1}^N \left(K_{d,n}^*(0)-K_{d,n}^*(s) - s \sum\limits_{m=1}^N \gamma_{d, m} w_{mn}^* \alpha_n^*(s)\right) 
\\ \begin{array}{cc}
{\text{s.t.}}&{ \sum\limits_{n=1}^N \left(K_{d,n}^*(0)-K_{d,n}^*(s) - s \sum\limits_{m=1}^N \gamma_{d, m} w_{mn}^* \alpha_n^*(s)\right)  \geq 0.}
\end{array}
\end{array}
\end{equation}
Problem (\ref{eq:InsuranceN}) is a nonlinear programming problem and it is challenging to find the analytical solution. We can leverage numerical methods to compute $s^*$ and obtain $T^*$ with (\ref{eq:DefenderMaxN}). We can then find $p_{d,n}^*$ and $p_{a,n}^*$ by plugging $s^*$ into Proposition \ref{pro:FlipItEqu}, and further obtain the solution of the \texttt{FlipIn-C} game defined in Definition \ref{def:BiLevelN}.

\subsection{Semi-homogeneous Case}
Problem (\ref{eq:InsuranceN}) can be directly solved in a semi-homogeneous case following similar steps in the previous section. In this semi-homogeneous case, we consider that all players in one party are homogeneous with the same parameters, i.e., $c_{d,n} = c_d$, $\gamma_{d, n} = \gamma_d$, $c_{a, n} = c_a$, and $\gamma_{a, n} = \gamma_a$ for $n\in\mathcal{N}$. Note that the network can be heterogeneous, i.e., each node may have a different number of neighbors with different $w_{mn}$. 

Recall (\ref{eq:DefenderGamma}), we have $\tilde{\gamma}_{d,n} = \sum_{m=1}^N \gamma_{d, m} w_{mn}^* = \gamma_{d}\sum_{m=1}^N  w_{mn}^* = \frac{\gamma_{d}}{1-\eta}$ for $n\in\mathcal{N}$ from Remark \ref{rem:DefenderRSolve}. Since the equilibrium of the \texttt{L-FlipIt-C} game only depends on $c_{d,n}$, $\tilde{\gamma}_{d, n}$, $c_{a, n}$, $\gamma_{a,n}$, and $\tilde{s}_n$, which are same for each node, all nodes have the same results at the equilibrium, i.e., $p_{d,n}^* = p_d^*$, $p_{a,n}^* = p_a^*$, $\alpha_{n}^* = \alpha^*$, and $J_{d,n}(p_{d,n}^*,p_{a,n}^*;s) = J_{d}(p_{d}^*,p_{a}^*;s)$. Thus, let $K_d^*(s)=J_{d}(p_{d}^*,p_{a}^*;s)$, the insurer's problem can be simplified into the following problem
\begin{equation}
\label{eq:InsuranceNH}
\begin{array}{l}
s^*=\arg\max\limits_{0<s\leq 1} \  K_{d}^*(0)-K_{d}^*(s) - s \frac{\gamma_{d}}{1-\eta} \alpha^*(s)
\\ \begin{array}{cc}
{\text{s.t.}}&{K_{d}^*(0)-K_{d}^*(s) - s \frac{\gamma_{d}}{1-\eta} \alpha^*(s)  \geq 0.}
\end{array}
\end{array} 
\end{equation}
Note that the premium $T^* = N(K_{d}^*(0)-K_{d}^*(s^*) )$, where $N$ is the number of nodes. We note that the structure of the games in this subsection is similar to the structure of the games in Defender-D, and we can obtain the equilibrium of the \texttt{FlipIn-C} game in this semi-homogeneous case using the results from Section \ref{sec:D}. 
\begin{corollary}
\label{cor:InsuranceNH}
The Nash equilibrium of the \texttt{FlipIn-C} game defined in Definition \ref{def:BiLeveln} of a semi-homogeneous case can be summarized into the following three cases:
\begin{itemize}
\item if $\frac{(1-\eta)\gamma_{a}c_{d}}{\gamma_{d}c_{a}}\geq 1$, the equilibrium is achieved at $s^* = \frac{1}{2}, T^* =  \frac{N\gamma_{d}}{2(1-\eta)}, p_{d}^* = \frac{ \gamma_{d}^2 c_{a}}{8(1-\eta)^2\gamma_{a}c_{d}^2},p_{a}^*= \frac{\gamma_{d}}{4(1-\eta)c_{d}}$;
\item if $\frac{1}{2} + \frac{\sqrt{2}}{4}\leq \frac{(1-\eta)\gamma_{a}c_{d}}{\gamma_{d}c_{a}}<1$, the equilibrium is achieved at $s^* = \frac{1}{2},  T^* =  \frac{Nc_{d}\gamma_{a}}{c_{a}}-\frac{N\gamma_{d}}{2(1-\eta)}, p_{d}^* =  \frac{ \gamma_{d}^2 c_{a}}{8(1-\eta)^2\gamma_{a}c_{d}^2},p_{a}^* =  \frac{\gamma_{d}}{4(1-\eta)c_{d}}$;
\item if $\frac{(1-\eta)\gamma_{a}c_{d}}{\gamma_{d}c_{a}} < \frac{1}{2} + \frac{\sqrt{2}}{4}$, the equilibrium does not exist. The defender and the attacker have $p_{d}^* =  \frac{\gamma_{a}}{2c_{a}}, p_{a}^*= \frac{(1-\eta)\gamma_{a}^2c_{d}}{2\gamma_{d}c_{a}^2}$.
\end{itemize} 
\end{corollary}
We can see that the equilibrium results of the semi-homogeneous \texttt{FlipIn-C} game do not depend on the network topology. We could also obtain the equilibrium results of the \texttt{FlipIn-D} games of Defender-D in this semi-homogeneous case by Proposition \ref{pro:EquBin}. Note that $\frac{(1-\eta)\gamma_{a}c_{d}}{\gamma_{d}c_{a}} < \frac{\gamma_{a}c_{d}}{\gamma_{d}w_{nn}^*c_{a}}$ as $w_{nn}^* < \frac{1}{1-\eta}$ from Remark \ref{rem:DefenderRSolve}, thus, the defender in Defender-C is less insurable than the defenders in Defender-D. Moreover, both the defender and the attacker act more frequently in Defender-C than in Defender-D when the defenders in both Defender-D and Defender-C are insurable. The defender defends at the same rate in Defender-D and in Defender-C while the attacker attacks more frequently in Defender-D than in Defender-C when the defenders in both Defender-D and Defender-C are not insurable. 

\section{Numerical Analysis}
\label{sec:Num}
In this section, we present three numerical experiments and compare the results in Defender-D and the results in Defender-C. In the first and second experiments, we consider homogeneous players and investigate the impacts of network topology. The first experiment compares the results of homogeneous networks with different levels of connectivity, while the second experiment compares the results of nodes with different numbers of neighbors in a heterogeneous network. In the last experiment, we consider heterogeneous players in a homogeneous network and compare the results of defenders with different cost parameters. 

These three experiments are inspired by real-world IoT applications. The first and second experiments consider homogeneous IoT devices, such as thermal controllers, surveillance cameras, and unmanned aerial vehicles (UAVs). It is important for both defenders and insurers to know how network topology affects the security of IoT networks. The third experiment considers heterogeneous IoT devices in a network. One example is that a smart home may contain laptops, wireless routers, smart speakers, cameras, and sweeping robots. Different devices may have distinct vulnerabilities and require different protection methods. Moreover, some devices, such as laptops and cameras, contain sensitive information of the household, and they may inflict higher losses on the defenders once they are compromised. Thus, it is also crucial to study cyber insurance on different devices in a network. 

\begin{figure}[]
\centering
\subfigure[]{\includegraphics[width=0.11\textwidth]{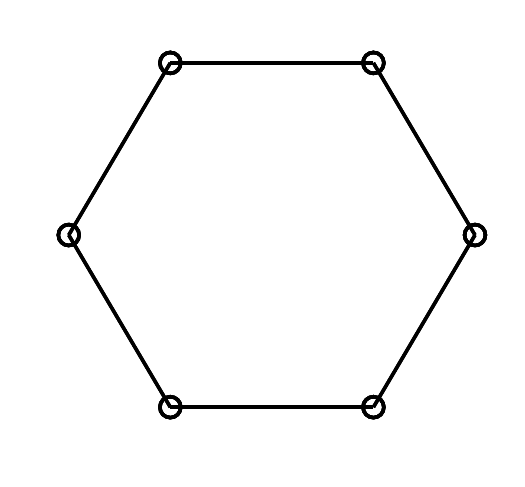}}
\subfigure[]{\includegraphics[width=0.11\textwidth]{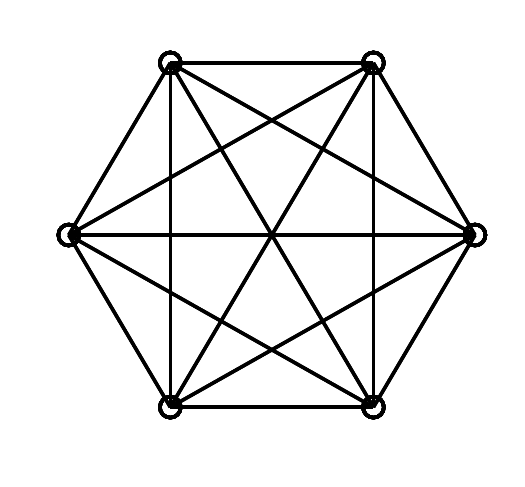}}
\subfigure[]{\includegraphics[width=0.11\textwidth]{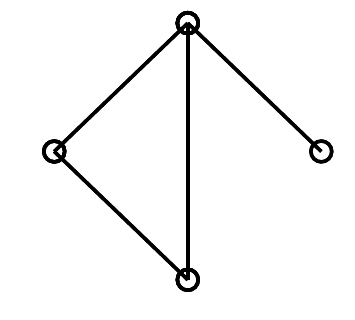}}
\subfigure[]{\includegraphics[width=0.11\textwidth]{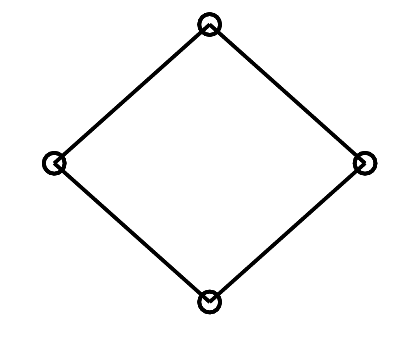}}
\caption{Networks.}
\label{fig:ExNetworks}
\end{figure}

\begin{figure}[]
\centering
\subfigure{\includegraphics[width=0.24\textwidth]{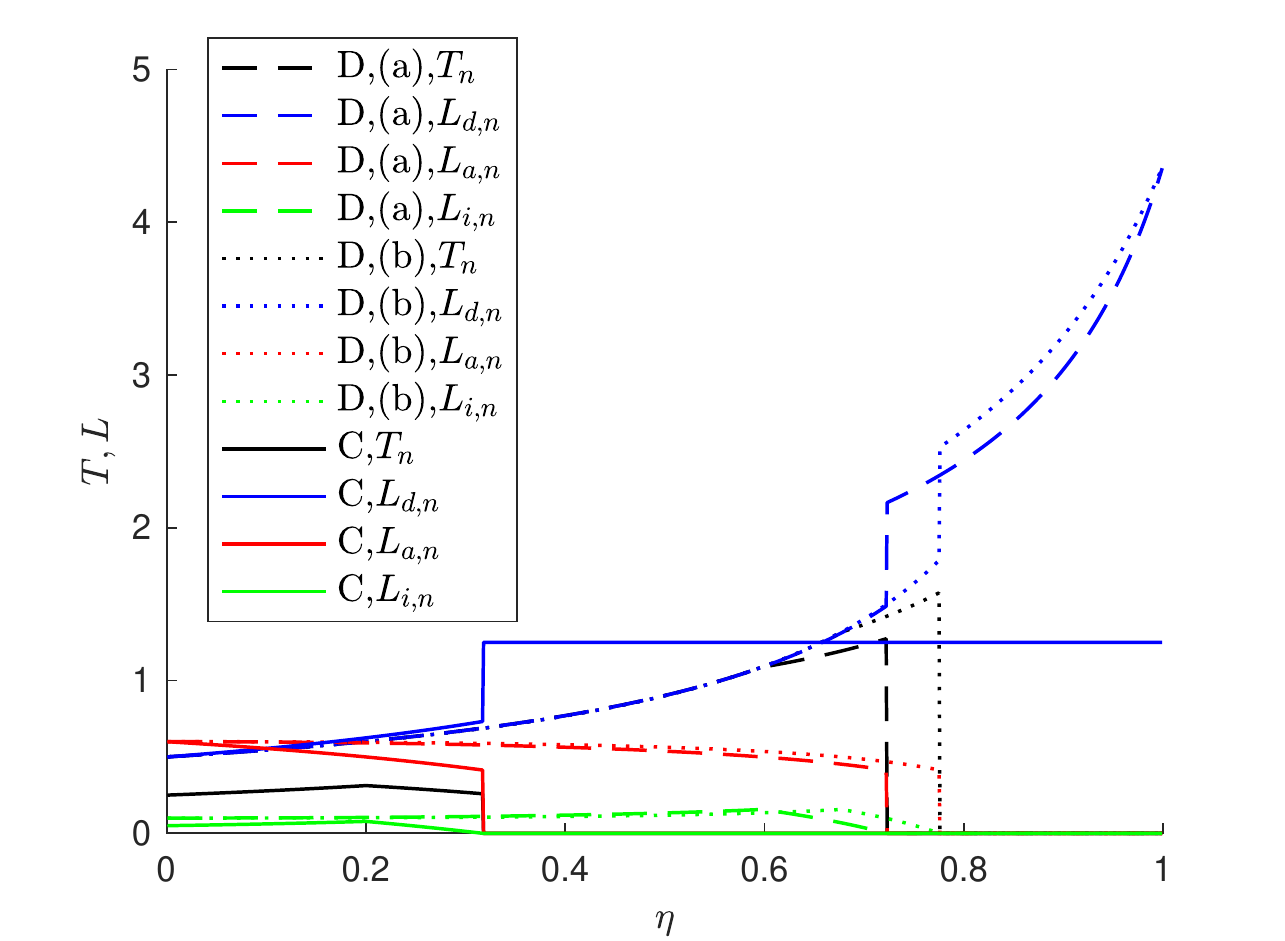}}
\subfigure{\includegraphics[width=0.24\textwidth]{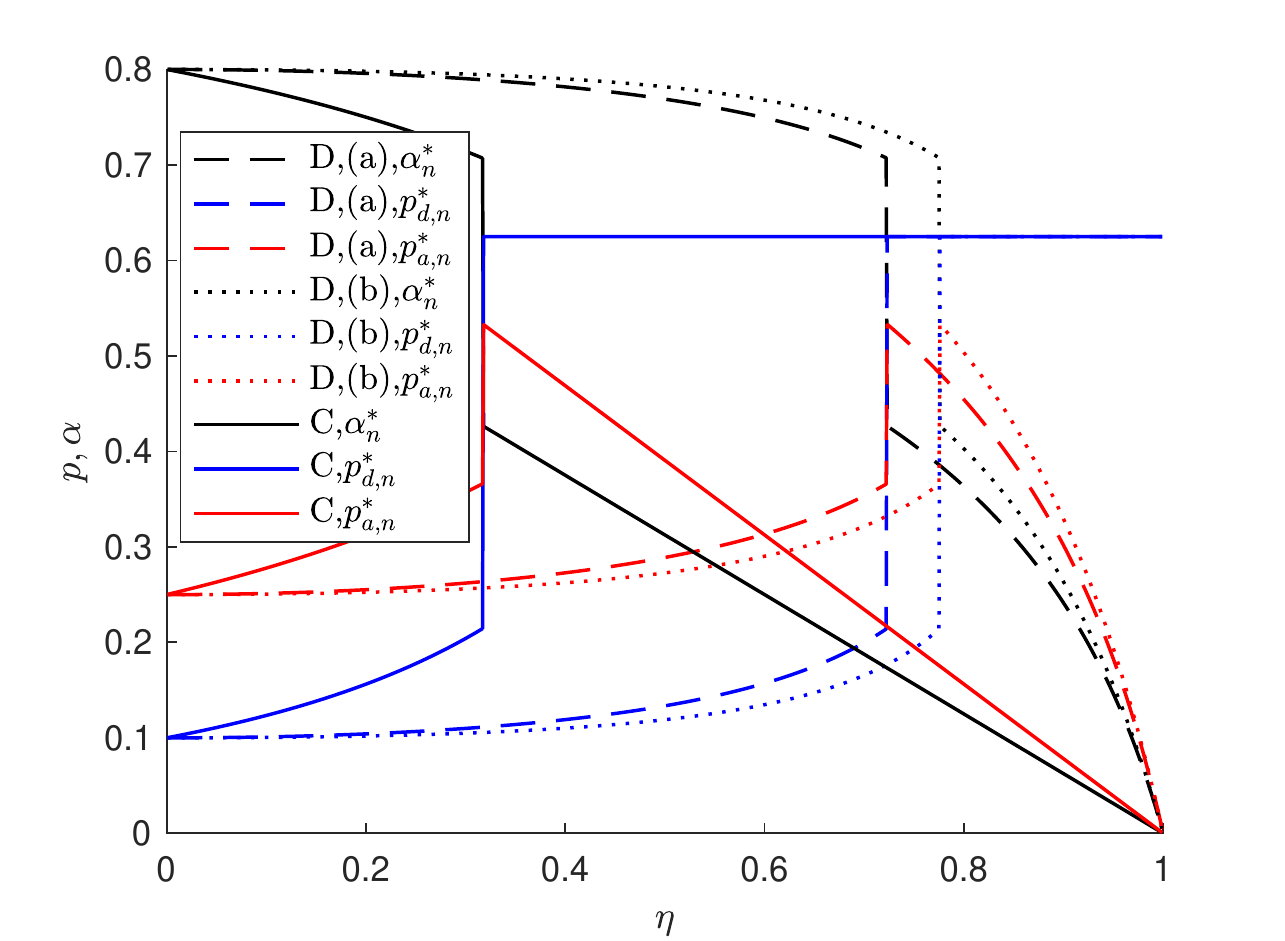}}
\caption{{Numerical results of homogeneous networks with homogeneous players in Fig. \ref{fig:ExNetworks}(ab). The x-axis is the network discount ratio $\eta$. ``D" and ``C" represent Defender-D and Defender-C, respectively. ``(a)" and ``(b)" represent networks (a) and (b) in Fig. \ref{fig:ExNetworks}, respectively. Note that each node in (a) and (b) has $2$ and $5$ neighbors, respectively. All nodes in both graphs satisfy $\gamma_{d,n}=1.0$, $ \gamma_{a,n} = 1.0$, $c_{d,n} = 1.0$, and $c_{a, n} = 0.8$. Each link in (a) and (b) satisfies $w_{nm} = 0.5$ and $0.2$, respectively. Different line styles indicate results in different networks while different line colors indicate different variables. Note that there are only one insurer and one defender for the case Defender-C, and we plot $T_n = T/N$ and $L_{d,n} = L_d/N$ instead of $T$ and $L_d$ to compare with the results in Defender-D.} }
\label{fig:ExHom}
\end{figure}

All the results in three experiments have been plotted with respect to the network discount ratio $\eta$ as shown in Figs. \ref{fig:ExHom}, \ref{fig:ExSemi}, and \ref{fig:ExHet}. A larger $\eta$ indicates that the network is strongly connected and an attacker can inflict larger losses on the neighboring IoT devices. We plot the losses of defenders as $L_{d,n}$ or $L_d$, utilities of attackers as $L_{a, n}$ or $L_a$, and profits of insurers as $L_{i,n}$ or $L_i$. Note that $L_{d,n}$ in Defender-D is computed through (\ref{eq:DefendernMin}) instead of (\ref{eq:DefendernMinSimplified}) as there are also losses caused by the attackers in neighboring nodes. In all figures, $T = 0$ or $T_n = 0$ indicate that the defender is not insurable and there exists no effective insurance contract. In the first and second experiments, the coverage level of the optimal insurance contract is $\frac{1}{2}$ in both Defender-C and Defender-D when the defenders are insurable. 

{We have two important observations from all experiments. We can see that the premiums in both Defender-D and Defender-C have an upward trend with the increase of $\eta$, which indicates that defenders are required to pay higher premiums on strongly connected networks. However, when $\eta$ is too large, the premiums drop to $0$,  i.e., the defenders are not insurable. As a result, we can conclude that the network effect decreases the insurability of defenders. The insurers should either charge higher premiums or provide no insurance to defenders while the defenders should improve local protections instead of purchasing insurance on strongly connected networks. }

{We can also see from all experiments that when $\eta$ is small and defenders are insurable, the defender's total loss in Defender-C is higher than the defenders' global loss in Defender-D, however, when the $\eta$ is large and defenders are not insurable, the defender's total loss in Defender-C is lower than the defenders' global loss in Defender-D. This phenomenon provides guidance for IoT defenders to decide between centralized management or decentralized management: for weakly connected networks, decentralized management is better than centralized management and each defender should monitor his or her own device; for strongly connected networks, centralized management outperforms decentralized management and a global defender should in charge of all devices. }

\begin{figure}[]
\centering
\subfigure[]{\includegraphics[width=0.24\textwidth]{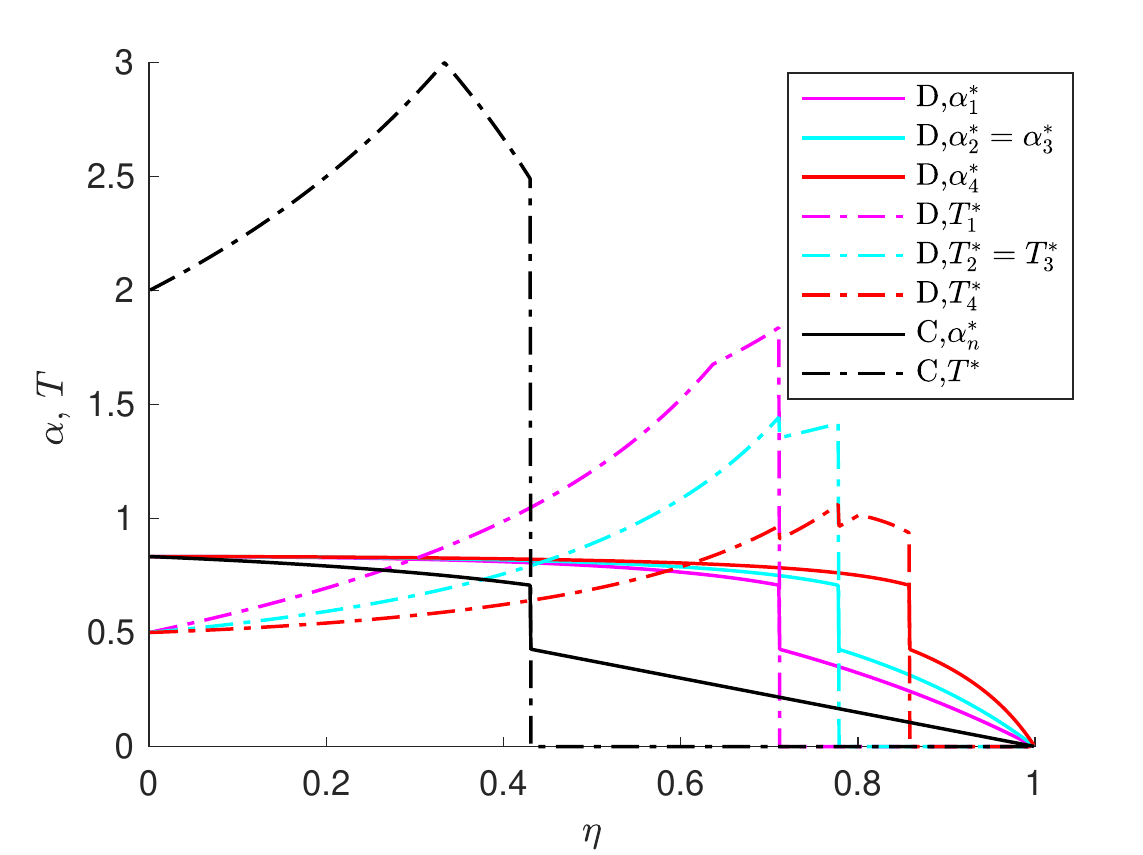}}
\subfigure[]{\includegraphics[width=0.24\textwidth]{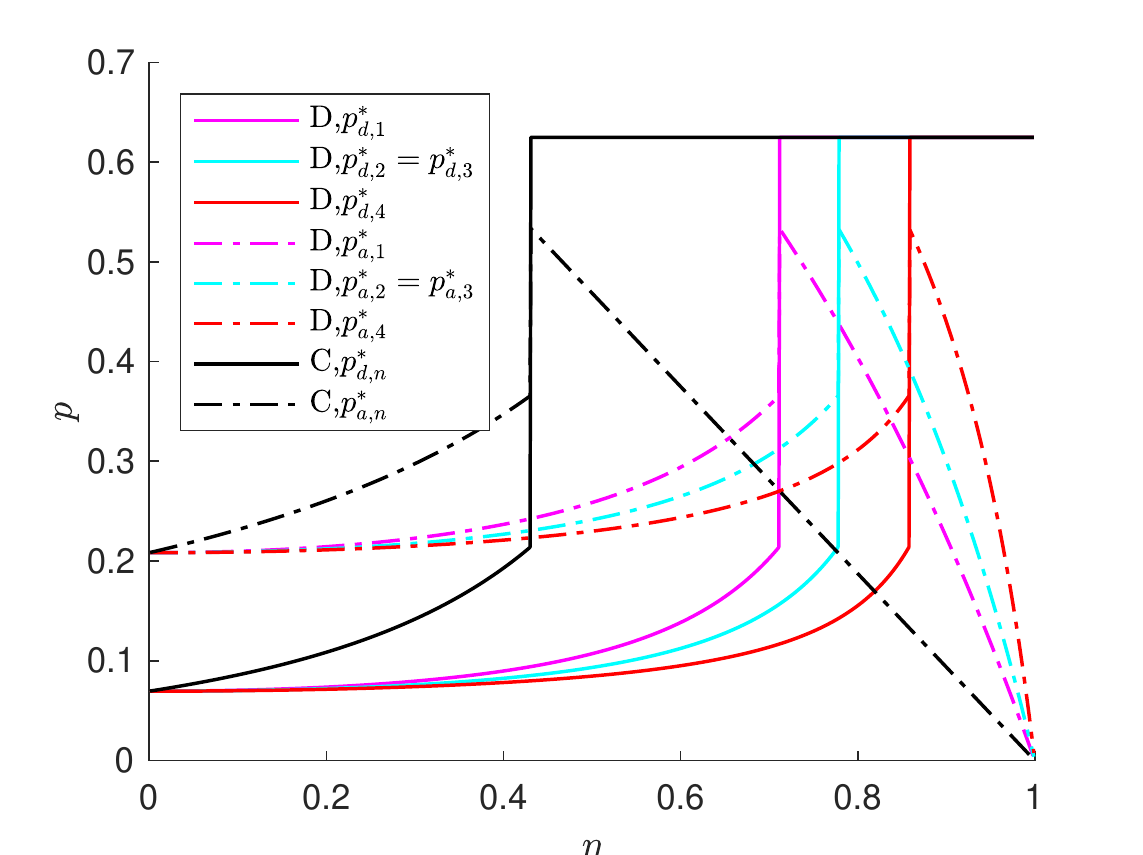}}
\subfigure[]{\includegraphics[width=0.24\textwidth]{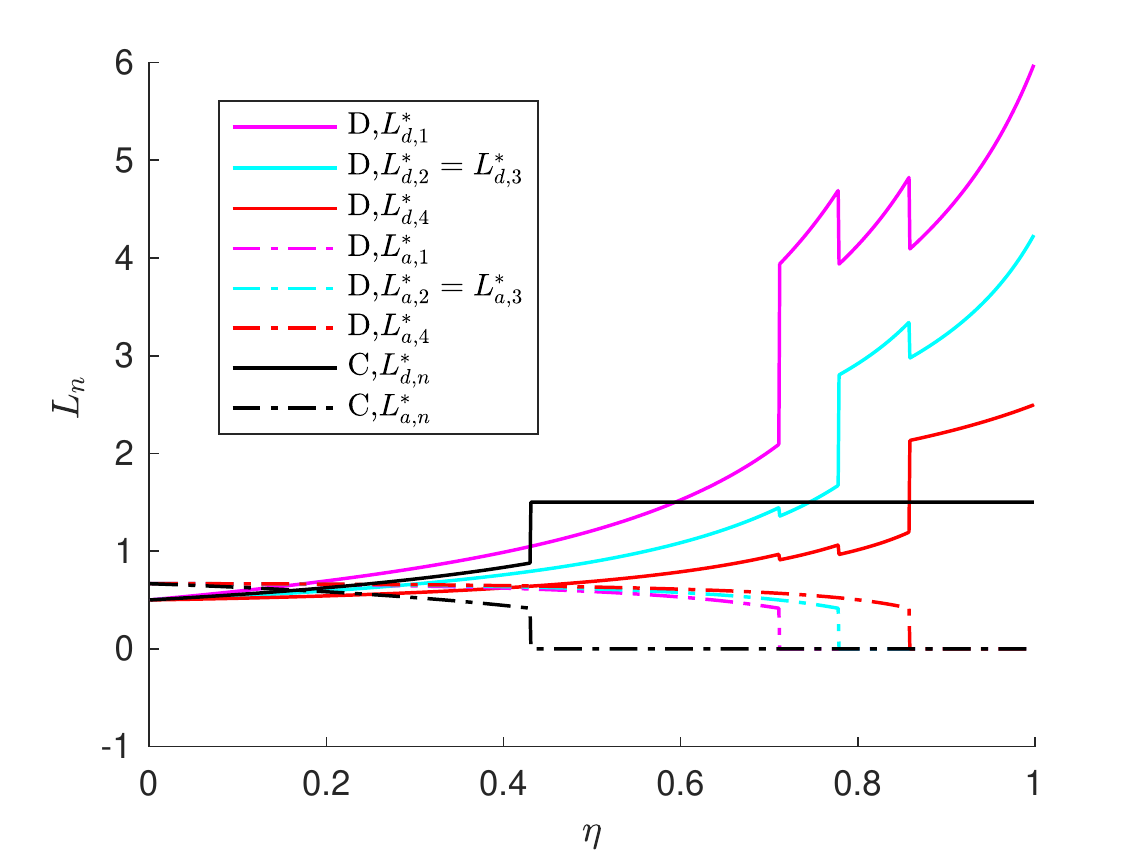}}
\subfigure[]{\includegraphics[width=0.24\textwidth]{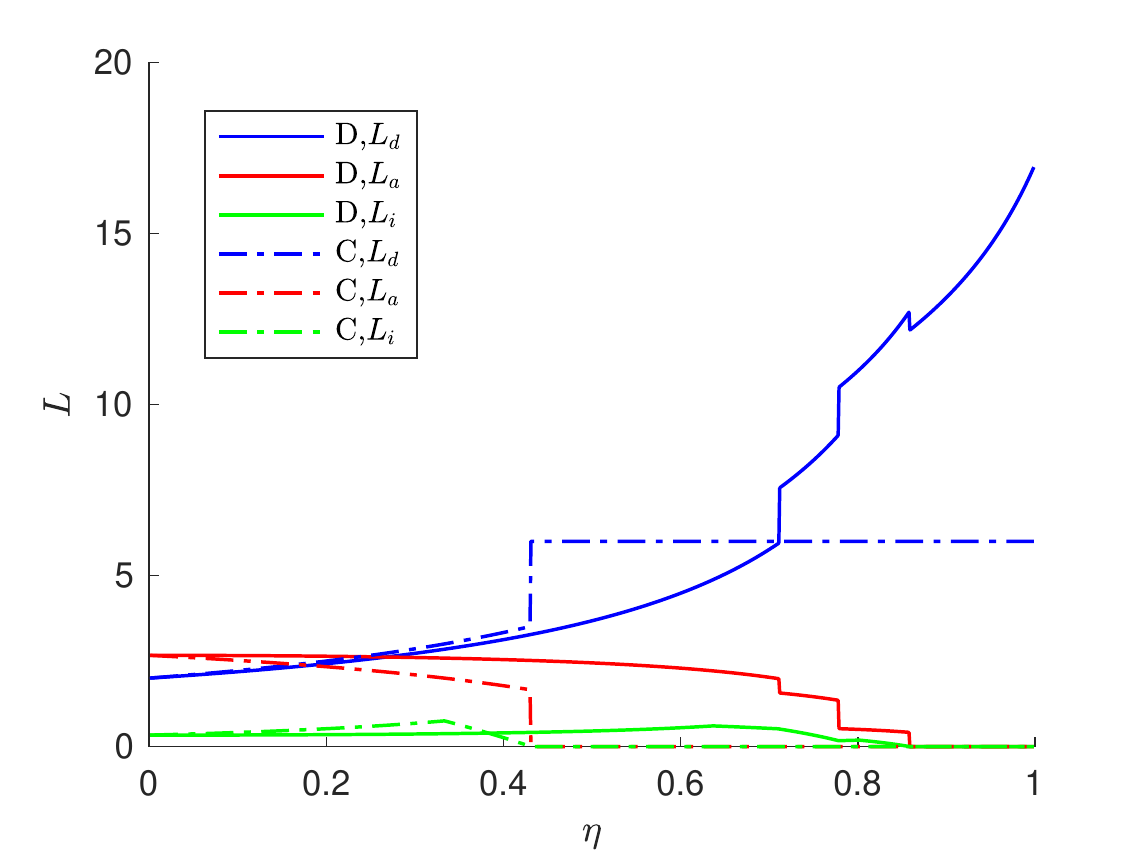}}
\caption{{Numerical results of a heterogeneous network with homogeneous IoT devices in Fig. \ref{fig:ExNetworks}(c). The x-axis is the network discount ratio $\eta$. ``D" and ``C" represent Defender-D and Defender-C, respectively. Note that nodes $1$, $2$, $3$, and $4$ have $3$, $2$, $2$, and $1$ neighbors, respectively, thus, $w_{12}=w_{13}=w_{14}=0.3333$, $w_{21} = w_{23} = w_{31} = w_{32} = 0.5$, and $w_{41} = 1$. All nodes satisfy $\gamma_{d,n}=1.0$, $ \gamma_{a,n} = 1.0$, $c_{d,n} = 1.2$, and $c_{a, n} = 0.8$. Different line colors indicate results of different nodes while different line styles indicate different variables. Note that we plot the global results of Defender-D in Fig. \ref{fig:ExSemi}(d) to compare with the results of Defender-C.}}
\label{fig:ExSemi}
\end{figure}

The results of the first experiment are presented in Fig. \ref{fig:ExHom}. Since all players are homogeneous and all networks are homogeneous, we only plot the results of one node for either network. Since all players are homogeneous, we can achieve the same results for both networks in Defender-C as discussed in Section \ref{sec:C}.A.; thus, we only plot the results of one network in Defender-C. Comparing the results of network (a) and network (b) in Defender-D, we can see that nodes in network (a) become uninsurable at a smaller $\eta$, which indicates that networks with lower connectivities are less insurable. Moreover, the global defender in Defender-C becomes uninsurable at a smaller $\eta$ than the defenders in Defender-D for both networks, which indicates that a defender who controls the whole network is less insurable than a defender who controls a single device.

The results of the second experiment are presented in Fig. \ref{fig:ExSemi}. Note that all nodes reach the same \texttt{L-FlipIt-C} equilibrium in Defender-C as discussed in Section \ref{sec:C}.A., and thus we only need to plot the results of one node for Defender-C. We also plot $L_d = \sum_{n\in\mathcal{N}}L_{d,n}$, $L_a = \sum_{n\in\mathcal{N}}L_{a,n}$, and $L_i = \sum_{n\in\mathcal{N}}L_{i,n}$ for Defender-D in Fig.\ref{fig:ExSemi}(d). Comparing the results of different nodes in Defender-D, we note that node 1 has a higher premium than nodes 2-4 and node 1 becomes uninsurable at a smaller $\eta$ from Fig. \ref{fig:ExSemi}(a). Thus, nodes with more neighbors are less insurable and they should be charged with higher premiums. Moreover, the defender at node 1 has a higher loss than the defenders at nodes 2-4, which indicates that nodes with more neighbors are more vulnerable. 

The results of the third experiment are presented in Fig. \ref{fig:ExHet}. Note that the defenders have different cost parameters in different nodes, thus, we need to solve problem (\ref{eq:InsuranceN}) with numerical methods to find the optimal insurance contracts in Defender-C. Comparing the results of different nodes in Defender-D, we can see that node 1 is always not insurable and node 2 becomes uninsurable at a smaller $\eta$ than nodes 3-4, which indicates that a defender who has a lower cost to protect his or her device is less insurable. Different from the first and second experiments, the defender switches between insurable statuses and uninsurable statuses with the increase of $\eta$ in Defender-C, and the coverage level is not $\frac{1}{2}$. We can see that both the coverage level and the premium increase with the increase of $\eta$ when the defender is at one insurable status. 

\begin{figure}[]
\centering
\subfigure{\includegraphics[width=0.24\textwidth]{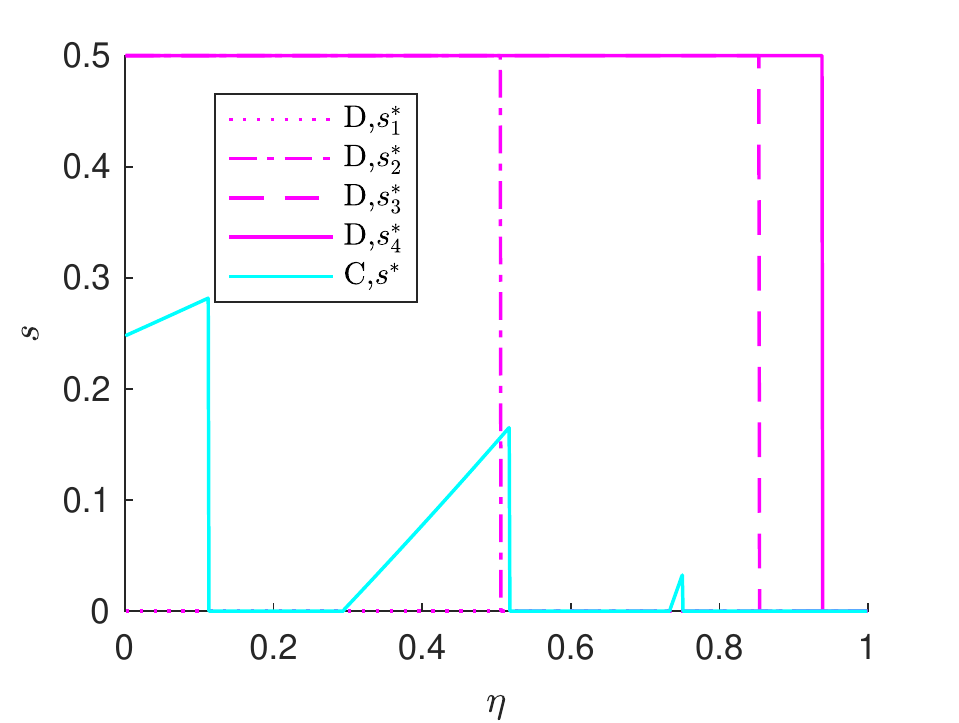}}
\subfigure{\includegraphics[width=0.24\textwidth]{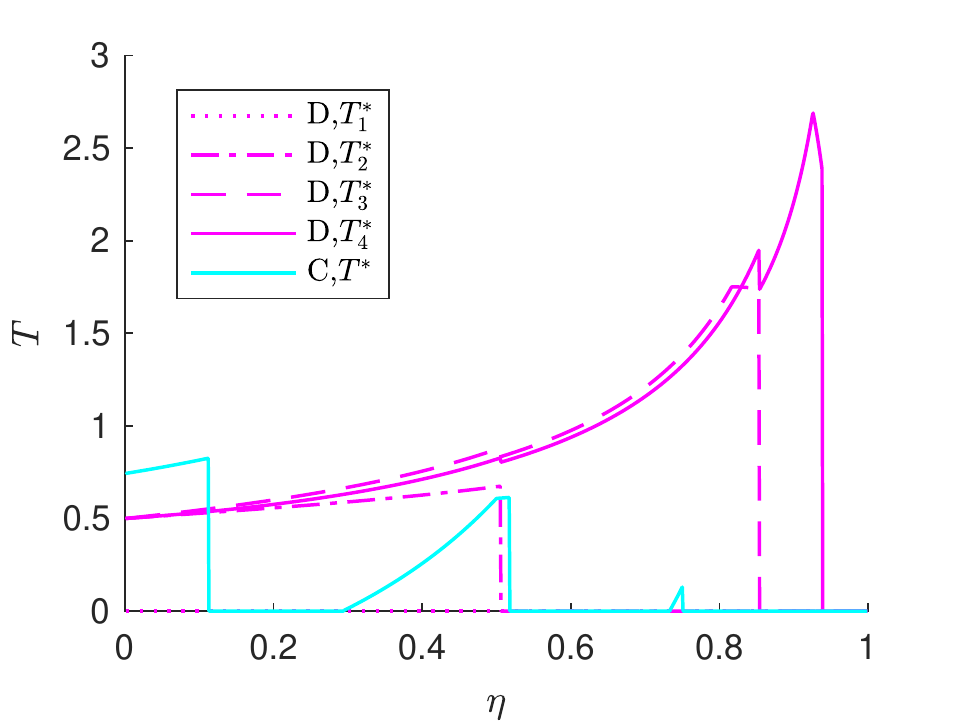}}
\subfigure{\includegraphics[width=0.24\textwidth]{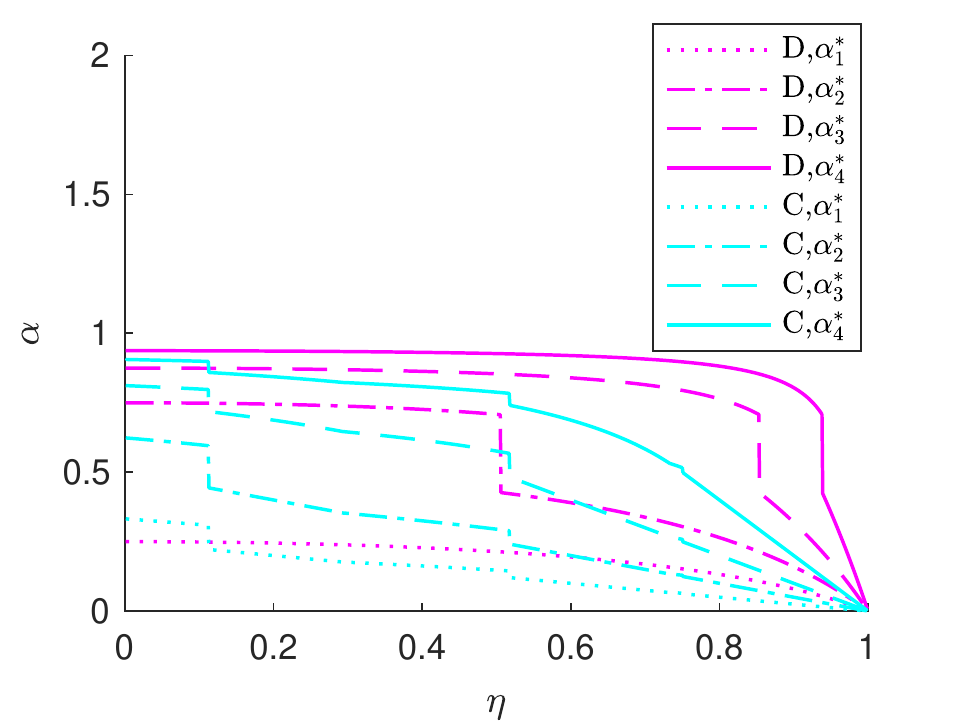}}
\subfigure{\includegraphics[width=0.24\textwidth]{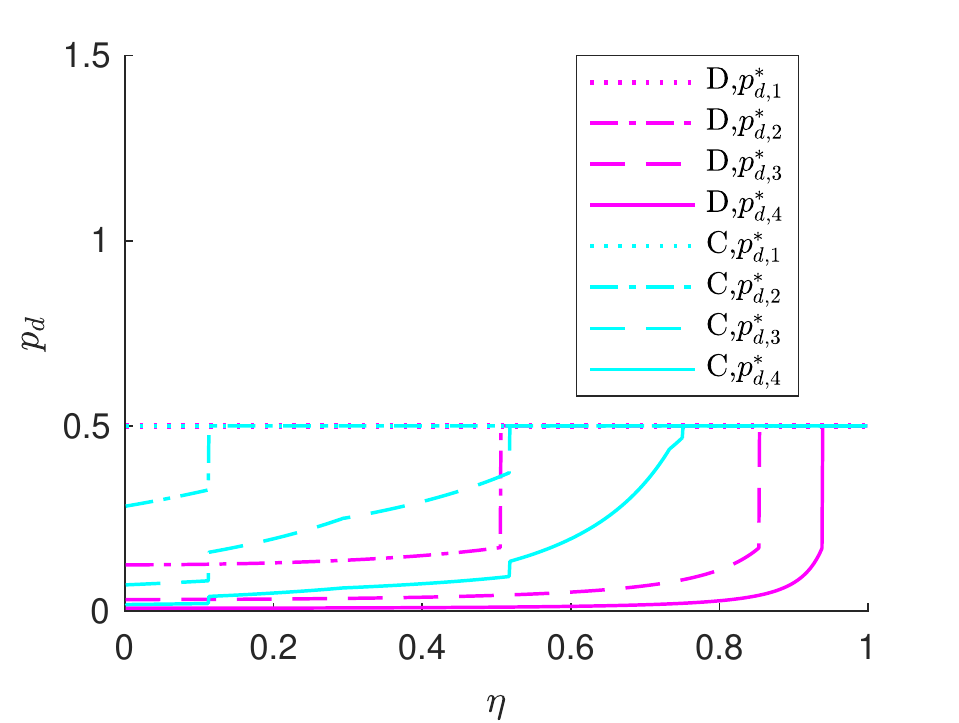}}
\subfigure{\includegraphics[width=0.24\textwidth]{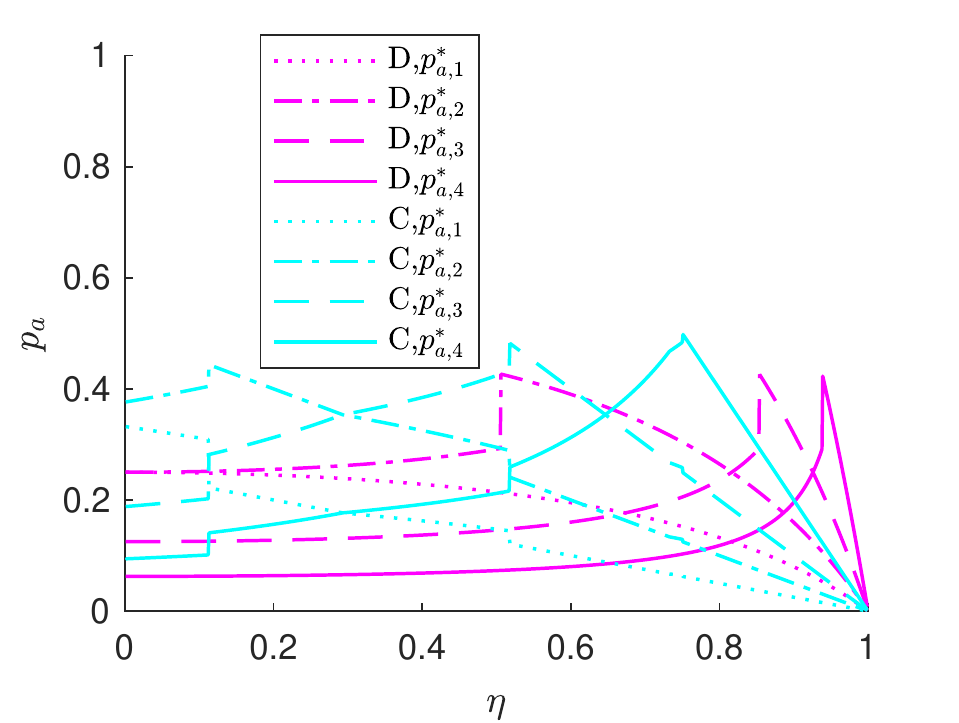}}
\subfigure{\includegraphics[width=0.24\textwidth]{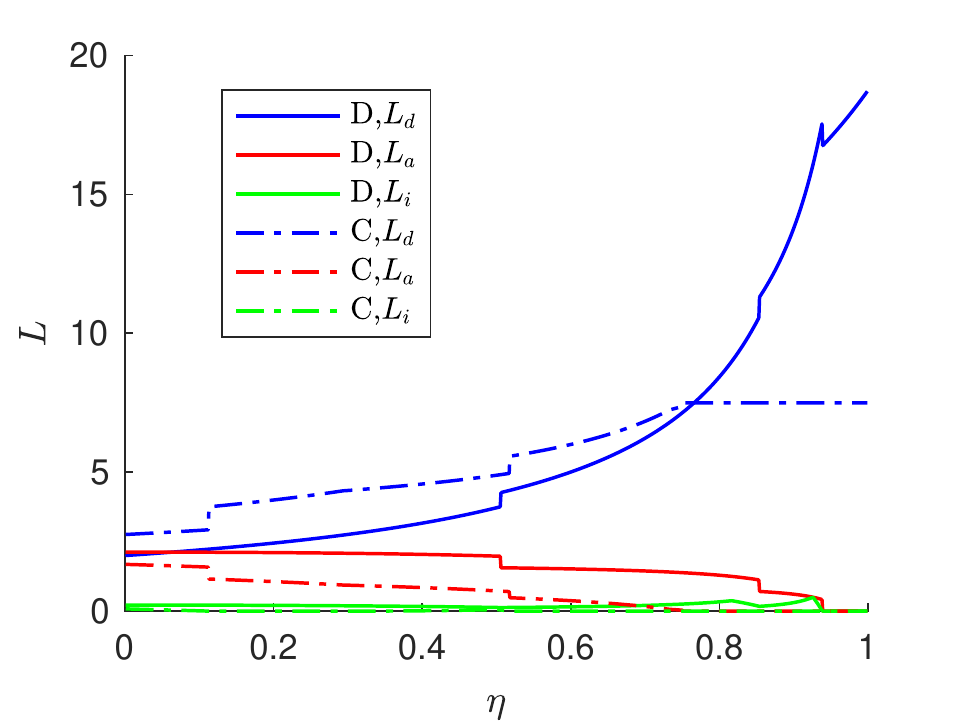}}
\caption{{Numerical results of a homogeneous network with heterogeneous IoT devices in Fig. \ref{fig:ExNetworks}(d). ``D" and ``C" represent Defender-D and Defender-C, respectively. All nodes satisfy $\gamma_{d,n}=1.0$, $ \gamma_{a,n} = 1.0$, and $c_{a, n} = 1.0$; all links satisfy $w_{nm} = 0.5$. Note that $c_{d,1} = 0.5$, $c_{d,2} = 1$, $c_{d,3} = 2$, and $c_{d,4} = 4$. Note that we plot the global results of Defender-D in the last subfigure of Fig. \ref{fig:ExHet} to compare with the results of Defender-C.}}
\label{fig:ExHet}
\end{figure}

\section{Conclusion}
\label{sec:Con}
In this paper, we have established the framework of \texttt{FlipIn} by composing \texttt{FlipIt} games and principal-agent problems to describe the complex interactions among defenders, attackers, and insurers over IoT networks. The framework has provided a theoretical underpinning for the quantitative assessment of cyber risks, the development of cross-layer defense mechanisms, and the design of cyber insurance policies. Through the analysis of the composed games, we have investigated the Peltzman effect of IoT owners and studied the fundamental concept of insurability. We have completely characterized the optimal insurance contracts for the case with a network of distributed defenders and the case with a centralized defender over a semi-homogeneous network. It has been shown that the optimal incentive-compatible insurance contract is to cover half of the defender's losses. Observations from numerical experiments have provided design guidelines and insights for designing policies for security management. There exists a nonlinear relationship between the level of connectivity and insurability. Nodes with low insurability need to invest in local cyber defense instead of counting on cyber insurance. One of the future directions would be the investigation of the dynamic \texttt{FlipIn} framework with partial observations that explores the optimal contract design under the time-varying cyber risks with imperfect measurements.

\section*{Appendix A Proof of Proposition \ref{pro:Insurability}}
From the constraint in (\ref{eq:Insurancen<}), we have
\[\begin{array}{l}
\frac{\gamma_{a, n}c_{d, n}}{c_{a, n}}-\gamma_{d, n} w_{nn}^* + \frac{ \gamma_{d, n} ^2w_{nn}^{*2} c_{a, n}}{2\gamma_{a, n} c_{d, n}}(1-s_n)s_n  \geq 0
\\ \Leftrightarrow
 \left(s_n - \frac{1}{2}\right)^2  \leq 2\left(\delta_n - \frac{1}{2}  \right)^2 - \frac{1}{4},
\end{array}  \]
where $\delta_n= \frac{\gamma_{a, n} c_{d, n}}{\gamma_{d, n} w_{nn}^{*} c_{a, n}}$ has been introduced to simplify the representations in this proof. Note that $\gamma_{a, n}c_{d, n} < \gamma_{d, n}w_{nn}^*c_{a, n}$ indicates $\delta_n < 1$ and $1-\frac{\gamma_{a, n}c_{d, n}}{\gamma_{d, n}w_{nn}^*c_{a, n}} \leq s_n \leq 1$ indicates $1-\delta_n \leq s_n \leq 1 $.

There exists $s_n$ only when $2\left(\delta_n  - \frac{1}{2}  \right)^2 - \frac{1}{4}  \geq 0$. As a result, $s_n$ is not feasible if $\frac{1}{2} -\frac{\sqrt{2}}{4}< \delta_n < \frac{1}{2} + \frac{\sqrt{2}}{4}$. 

We can further obtain that $s_n$ should satisfy 
\[\begin{array}{l}
\frac{1}{2} - \sqrt{2\left(\delta_n - \frac{1}{2}  \right)^2 - \frac{1}{4}} \leq  s_n  \leq \frac{1}{2} + \sqrt{2\left(\delta_n  - \frac{1}{2}  \right)^2 - \frac{1}{4}}.
\end{array}\]
Note that $s_n$ should also satisfy $1-\delta_n \leq s_n \leq 1$. Thus, $s_n$ is feasible only when 
\[ \begin{array}{l}
   \frac{1}{2}-\delta_n  \leq  \sqrt{2\left(\delta_n  - \frac{1}{2}  \right)^2 - \frac{1}{4}}.
\end{array}\]

If $0 < \delta_n \leq \frac{1}{2} -\frac{\sqrt{2}}{4}$, we have that $\left(\frac{1}{2} -\delta_n  \right)^2 \leq 2\left(\frac{1}{2} -\delta_n  \right)^2 - \frac{1}{4}$. Thus, $\delta_n$ should satisfy $\delta_n \geq 1$ or $\delta_n \leq 0$, which contradicts to $0 < \delta_n \leq \frac{1}{2} -\frac{\sqrt{2}}{4}$. As a result, $s_n$ is not feasible if $0 < \delta_n \leq \frac{1}{2} -\frac{\sqrt{2}}{4}$.

If $\frac{1}{2} + \frac{\sqrt{2}}{4}\leq \delta_n<1$, we have that $\frac{1}{2} - \delta_n < 0 \leq  \sqrt{2\left(\delta_n  - \frac{1}{2}  \right)^2 - \frac{1}{4}}$. As a result, $s_n$ is feasible if $\frac{1}{2} + \frac{\sqrt{2}}{4}\leq \delta_n<1$.

After summarizing everything, we can obtain that $s_n$ is not feasible if $0 < \delta_n < \frac{1}{2} + \frac{\sqrt{2}}{4}$, thus, Proposition \ref{pro:Insurability} holds. 

{
\bibliographystyle{ieeetr}
\bibliography{DraftRZ.bib}

\begin{thebibliography}{10}

\bibitem{atzori2010internet}
L.~Atzori, A.~Iera, and G.~Morabito, ``The internet of things: A survey,'' {\em
  Computer networks}, vol.~54, no.~15, pp.~2787--2805, 2010.

\bibitem{gubbi2013internet}
J.~Gubbi, R.~Buyya, S.~Marusic, and M.~Palaniswami, ``Internet of things (iot):
  A vision, architectural elements, and future directions,'' {\em Future
  generation computer systems}, vol.~29, no.~7, pp.~1645--1660, 2013.

\bibitem{weber2010internet}
R.~H. Weber, ``Internet of things--new security and privacy challenges,'' {\em
  Computer law \& security review}, vol.~26, no.~1, pp.~23--30, 2010.

\bibitem{suo2012security}
H.~Suo, J.~Wan, C.~Zou, and J.~Liu, ``Security in the internet of things: a
  review,'' in {\em Computer Science and Electronics Engineering (ICCSEE), 2012
  international conference on}, vol.~3, pp.~648--651, IEEE, 2012.

\bibitem{zhao2013survey}
K.~Zhao and L.~Ge, ``A survey on the internet of things security,'' in {\em
  Computational Intelligence and Security (CIS), 2013 9th International
  Conference on}, pp.~663--667, IEEE, 2013.

\bibitem{jing2014security}
Q.~Jing, A.~V. Vasilakos, J.~Wan, J.~Lu, and D.~Qiu, ``Security of the internet
  of things: perspectives and challenges,'' {\em Wireless Networks}, vol.~20,
  no.~8, pp.~2481--2501, 2014.

\bibitem{sicari2015security}
S.~Sicari, A.~Rizzardi, L.~A. Grieco, and A.~Coen-Porisini, ``Security, privacy
  and trust in internet of things: The road ahead,'' {\em Computer networks},
  vol.~76, pp.~146--164, 2015.

\bibitem{antonakakis2017understanding}
M.~Antonakakis, T.~April, M.~Bailey, M.~Bernhard, E.~Bursztein, J.~Cochran,
  Z.~Durumeric, J.~A. Halderman, L.~Invernizzi, M.~Kallitsis, {\em et~al.},
  ``Understanding the mirai botnet,'' in {\em USENIX Security Symposium},
  pp.~1092--1110, 2017.

\bibitem{kolias2017ddos}
C.~Kolias, G.~Kambourakis, A.~Stavrou, and J.~Voas, ``Ddos in the iot: Mirai
  and other botnets,'' {\em Computer}, vol.~50, no.~7, pp.~80--84, 2017.

\bibitem{tankard2011advanced}
C.~Tankard, ``Advanced persistent threats and how to monitor and deter them,''
  {\em Network security}, vol.~2011, no.~8, pp.~16--19, 2011.

\bibitem{cole2012advanced}
E.~Cole, {\em Advanced persistent threat: understanding the danger and how to
  protect your organization}.
\newblock Newnes, 2012.

\bibitem{abomhara2015cyber}
M.~Abomhara and G.~M. K{\o}ien, ``Cyber security and the internet of things:
  vulnerabilities, threats, intruders and attacks,'' {\em Journal of Cyber
  Security}, vol.~4, no.~1, pp.~65--88, 2015.

\bibitem{hassanzadeh2015towards}
A.~Hassanzadeh, S.~Modi, and S.~Mulchandani, ``Towards effective security
  control assignment in the industrial internet of things,'' in {\em Internet
  of Things (WF-IoT), 2015 IEEE 2nd World Forum on}, pp.~795--800, IEEE, 2015.

\bibitem{hu2017defense}
Q.~Hu, S.~Lv, Z.~Shi, L.~Sun, and L.~Xiao, ``Defense against advanced
  persistent threats with expert system for internet of things,'' in {\em
  International Conference on Wireless Algorithms, Systems, and Applications},
  pp.~326--337, Springer, 2017.

\bibitem{langner2011stuxnet}
R.~Langner, ``Stuxnet: Dissecting a cyberwarfare weapon,'' {\em IEEE Security
  \& Privacy}, vol.~9, no.~3, pp.~49--51, 2011.

\bibitem{kushner2013real}
D.~Kushner, ``The real story of stuxnet,'' {\em ieee Spectrum}, vol.~3, no.~50,
  pp.~48--53, 2013.

\bibitem{zhang2014iot}
Z.-K. Zhang, M.~C.~Y. Cho, C.-W. Wang, C.-W. Hsu, C.-K. Chen, and S.~Shieh,
  ``Iot security: ongoing challenges and research opportunities,'' in {\em
  Service-Oriented Computing and Applications (SOCA), 2014 IEEE 7th
  International Conference on}, pp.~230--234, IEEE, 2014.

\bibitem{mahmoud2015internet}
R.~Mahmoud, T.~Yousuf, F.~Aloul, and I.~Zualkernan, ``Internet of things (iot)
  security: Current status, challenges and prospective measures,'' in {\em
  Internet Technology and Secured Transactions (ICITST), 2015 10th
  International Conference for}, pp.~336--341, IEEE, 2015.

\bibitem{majuca2006evolution}
R.~P. Majuca, W.~Yurcik, and J.~P. Kesan, ``The evolution of cyberinsurance,''
  {\em arXiv preprint cs/0601020}, 2006.

\bibitem{bohme2010modeling}
R.~B{\"o}hme, G.~Schwartz, {\em et~al.}, ``Modeling cyber-insurance: Towards a
  unifying framework.,'' in {\em WEIS}, 2010.

\bibitem{marotta2017cyber}
A.~Marotta, F.~Martinelli, S.~Nanni, A.~Orlando, and A.~Yautsiukhin,
  ``Cyber-insurance survey,'' {\em Computer Science Review}, vol.~24,
  pp.~35--61, 2017.

\bibitem{miura2008security}
R.~A. Miura-Ko, B.~Yolken, N.~Bambos, and J.~Mitchell, ``Security investment
  games of interdependent organizations,'' in {\em Communication, Control, and
  Computing, 2008 46th Annual Allerton Conference on}, pp.~252--260, IEEE,
  2008.

\bibitem{nguyen2009stochastic}
K.~C. Nguyen, T.~Alpcan, and T.~Basar, ``Stochastic games for security in
  networks with interdependent nodes,'' in {\em Game Theory for Networks, 2009.
  GameNets' 09. International Conference on}, pp.~697--703, IEEE, 2009.

\bibitem{alpcan2010network}
T.~Alpcan and T.~Ba{\c{s}}ar, {\em Network security: A decision and
  game-theoretic approach}.
\newblock Cambridge University Press, 2010.

\bibitem{van2013flipit}
M.~Van~Dijk, A.~Juels, A.~Oprea, and R.~L. Rivest, ``Flipit: The game of
  “stealthy takeover”,'' {\em Journal of Cryptology}, vol.~26, no.~4,
  pp.~655--713, 2013.

\bibitem{bowers2012defending}
K.~D. Bowers, M.~Van~Dijk, R.~Griffin, A.~Juels, A.~Oprea, R.~L. Rivest, and
  N.~Triandopoulos, ``Defending against the unknown enemy: Applying flipit to
  system security,'' in {\em International Conference on Decision and Game
  Theory for Security}, pp.~248--263, Springer, 2012.

\bibitem{pawlick2015flip}
J.~Pawlick, S.~Farhang, and Q.~Zhu, ``Flip the cloud: cyber-physical signaling
  games in the presence of advanced persistent threats,'' in {\em International
  Conference on Decision and Game Theory for Security}, pp.~289--308, Springer,
  2015.

\bibitem{chen2017security}
J.~Chen and Q.~Zhu, ``Security as a service for cloud-enabled internet of
  controlled things under advanced persistent threats: a contract design
  approach,'' {\em IEEE Transactions on Information Forensics and Security},
  vol.~12, no.~11, pp.~2736--2750, 2017.

\bibitem{pal2014will}
R.~Pal, L.~Golubchik, K.~Psounis, and P.~Hui, ``Will cyber-insurance improve
  network security? a market analysis,'' in {\em INFOCOM, 2014 Proceedings
  IEEE}, pp.~235--243, IEEE, 2014.

\bibitem{khalili2017designing1}
M.~M. Khalili, P.~Naghizadeh, and M.~Liu, ``Designing cyber insurance policies:
  Mitigating moral hazard through security pre-screening,'' in {\em
  International Conference on Game Theory for Networks}, pp.~63--73, Springer,
  2017.

\bibitem{khalili2017designing2}
M.~M. Khalili, P.~Naghizadeh, and M.~Liu, ``Designing cyber insurance policies
  in the presence of security interdependence,'' in {\em Proceedings of the
  12th workshop on the Economics of Networks, Systems and Computation}, p.~7,
  ACM, 2017.

\bibitem{khalili2018designing}
M.~M. Khalili, P.~Naghizadeh, and M.~Liu, ``Designing cyber insurance policies:
  The role of pre-screening and security interdependence,'' {\em IEEE
  Transactions on Information Forensics and Security}, vol.~13, no.~9,
  pp.~2226--2239, 2018.

\bibitem{vakilinia2018coalitional}
I.~Vakilinia and S.~Sengupta, ``A coalitional cyber-insurance framework for a
  common platform,'' {\em IEEE Transactions on Information Forensics and
  Security}, 2018.

\bibitem{roy2010survey}
S.~Roy, C.~Ellis, S.~Shiva, D.~Dasgupta, V.~Shandilya, and Q.~Wu, ``A survey of
  game theory as applied to network security,'' in {\em System Sciences
  (HICSS), 2010 43rd Hawaii International Conference on}, pp.~1--10, IEEE,
  2010.

\bibitem{ferdowsi2017colonel}
A.~Ferdowsi, W.~Saad, B.~Maham, and N.~B. Mandayam, ``A colonel blotto game for
  interdependence-aware cyber-physical systems security in smart cities,'' in
  {\em Proceedings of the 2nd International Workshop on Science of Smart City
  Operations and Platforms Engineering}, pp.~7--12, ACM, 2017.

\bibitem{han2019game}
Z.~Han, D.~Niyato, W.~Saad, and T.~Ba{\c{s}}ar, {\em Game Theory for Next
  Generation Wireless and Communication Networks: Modeling, Analysis, and
  Design}.
\newblock Cambridge University Press, 2019.

\bibitem{li2015jamming}
Y.~Li, L.~Shi, P.~Cheng, J.~Chen, and D.~E. Quevedo, ``Jamming attacks on
  remote state estimation in cyber-physical systems: A game-theoretic
  approach,'' {\em IEEE Transactions on Automatic Control}, vol.~60, no.~10,
  pp.~2831--2836, 2015.

\bibitem{spyridopoulos2013gameC}
T.~Spyridopoulos, G.~Oikonomou, T.~Tryfonas, and M.~Ge, ``Game theoretic
  approach for cost-benefit analysis of malware proliferation prevention,'' in
  {\em IFIP International Information Security Conference}, pp.~28--41,
  Springer, 2013.

\bibitem{vadlamudi2016moving}
S.~G. Vadlamudi, S.~Sengupta, M.~Taguinod, Z.~Zhao, A.~Doup{\'e}, G.-J. Ahn,
  and S.~Kambhampati, ``Moving target defense for web applications using
  bayesian stackelberg games,'' in {\em Proceedings of the 2016 International
  Conference on Autonomous Agents \& Multiagent Systems}, pp.~1377--1378,
  International Foundation for Autonomous Agents and Multiagent Systems, 2016.

\bibitem{kiekintveld2015game}
C.~Kiekintveld, V.~Lis{\`y}, and R.~P{\'\i}bil, ``Game-theoretic foundations
  for the strategic use of honeypots in network security,'' in {\em Cyber
  Warfare}, pp.~81--101, Springer, 2015.

\bibitem{zhu2011stackelberg}
M.~Zhu and S.~Martinez, ``Stackelberg-game analysis of correlated attacks in
  cyber-physical systems,'' in {\em American Control Conference (ACC), 2011},
  pp.~4063--4068, IEEE, 2011.

\bibitem{carroll2011game}
T.~E. Carroll and D.~Grosu, ``A game theoretic investigation of deception in
  network security,'' {\em Security and Communication Networks}, vol.~4,
  no.~10, pp.~1162--1172, 2011.

\bibitem{teixeira2014security}
A.~Teixeira, G.~Dan, H.~Sandberg, R.~Berthier, R.~B. Bobba, and A.~Valdes,
  ``Security of smart distribution grids: Data integrity attacks on integrated
  volt/var control and countermeasures,'' in {\em American Control Conference
  (ACC), 2014}, pp.~4372--4378, IEEE, 2014.

\bibitem{wang2014mean}
Y.~Wang, F.~R. Yu, H.~Tang, and M.~Huang, ``A mean field game theoretic
  approach for security enhancements in mobile ad hoc networks,'' {\em IEEE
  transactions on wireless communications}, vol.~13, no.~3, pp.~1616--1627,
  2014.

\bibitem{zhu2015game}
Q.~Zhu and T.~Basar, ``Game-theoretic methods for robustness, security, and
  resilience of cyberphysical control systems: games-in-games principle for
  optimal cross-layer resilient control systems,'' {\em IEEE control systems},
  vol.~35, no.~1, pp.~46--65, 2015.

\bibitem{huang2018analysis}
L.~Huang and Q.~Zhu, ``Analysis and computation of adaptive defense strategies
  against advanced persistent threats for cyber-physical systems,'' in {\em
  International Conference on Decision and Game Theory for Security},
  pp.~205--226, Springer, 2018.

\bibitem{hu2015dynamic}
P.~Hu, H.~Li, H.~Fu, D.~Cansever, and P.~Mohapatra, ``Dynamic defense strategy
  against advanced persistent threat with insiders,'' in {\em Computer
  Communications (INFOCOM), 2015 IEEE Conference on}, pp.~747--755, IEEE, 2015.

\bibitem{xiao2017cloud}
L.~Xiao, D.~Xu, C.~Xie, N.~B. Mandayam, and H.~V. Poor, ``Cloud storage defense
  against advanced persistent threats: A prospect theoretic study,'' {\em IEEE
  Journal on Selected Areas in Communications}, vol.~35, no.~3, pp.~534--544,
  2017.

\bibitem{min2018defense}
M.~Min, L.~Xiao, C.~Xie, M.~Hajimirsadeghi, and N.~B. Mandayam, ``Defense
  against advanced persistent threats in dynamic cloud storage: A colonel
  blotto game approach,'' {\em IEEE Internet of Things Journal}, vol.~5, no.~6,
  pp.~4250--4261, 2018.

\bibitem{rass2016gadapt}
S.~Rass and Q.~Zhu, ``Gadapt: a sequential game-theoretic framework for
  designing defense-in-depth strategies against advanced persistent threats,''
  in {\em International Conference on Decision and Game Theory for Security},
  pp.~314--326, Springer, 2016.

\bibitem{manshaei2013game}
M.~H. Manshaei, Q.~Zhu, T.~Alpcan, T.~Bac{\c{s}}ar, and J.-P. Hubaux, ``Game
  theory meets network security and privacy,'' {\em ACM Computing Surveys
  (CSUR)}, vol.~45, no.~3, p.~25, 2013.

\bibitem{abuzainab2017dynamic}
N.~Abuzainab and W.~Saad, ``Dynamic connectivity game for adversarial internet
  of battlefield things systems,'' {\em IEEE Internet of Things Journal},
  vol.~5, no.~1, pp.~378--390, 2017.

\bibitem{hu2019dynamic}
Y.~Hu, A.~Sanjab, and W.~Saad, ``Dynamic psychological game theory for secure
  internet of battlefield things (iobt) systems,'' {\em IEEE Internet of Things
  Journal}, vol.~6, no.~2, pp.~3712--3726, 2019.

\bibitem{hamdi2014game}
M.~Hamdi and H.~Abie, ``Game-based adaptive security in the internet of things
  for ehealth,'' in {\em Communications (ICC), 2014 IEEE International
  Conference on}, pp.~920--925, IEEE, 2014.

\bibitem{pouryazdan2016game}
M.~Pouryazdan, C.~Fiandrino, B.~Kantarci, D.~Kliazovich, T.~Soyata, and
  P.~Bouvry, ``Game-theoretic recruitment of sensing service providers for
  trustworthy cloud-centric internet-of-things (iot) applications,'' in {\em
  IEEE Global Communications Conference (GLOBECOM) Workshops: Fifth
  International Workshop on Cloud Computing Systems, Networks, and Applications
  (CCSNA)}, 2016.

\bibitem{namvar2016jamming}
N.~Namvar, W.~Saad, N.~Bahadori, and B.~Kelley, ``Jamming in the internet of
  things: A game-theoretic perspective,'' in {\em Global Communications
  Conference (GLOBECOM), 2016 IEEE}, pp.~1--6, IEEE, 2016.

\bibitem{lee2018game}
S.~Lee, S.~Kim, K.~Choi, and T.~Shon, ``Game theory-based security
  vulnerability quantification for social internet of things,'' {\em Future
  Generation Computer Systems}, vol.~82, pp.~752--760, 2018.

\bibitem{spyridopoulos2013gameJ}
T.~Spyridopoulos, G.~Karanikas, T.~Tryfonas, and G.~Oikonomou, ``A game
  theoretic defence framework against dos/ddos cyber attacks,'' {\em Computers
  \& Security}, vol.~38, pp.~39--50, 2013.

\bibitem{laszka2014flipthem}
A.~Laszka, G.~Horvath, M.~Felegyhazi, and L.~Butty{\'a}n, ``Flipthem: Modeling
  targeted attacks with flipit for multiple resources,'' in {\em International
  Conference on Decision and Game Theory for Security}, pp.~175--194, Springer,
  2014.

\bibitem{zhang2015game}
M.~Zhang, Z.~Zheng, and N.~B. Shroff, ``A game theoretic model for defending
  against stealthy attacks with limited resources,'' in {\em International
  Conference on Decision and Game Theory for Security}, pp.~93--112, Springer,
  2015.

\bibitem{leslie2015threshold}
D.~Leslie, C.~Sherfield, and N.~P. Smart, ``Threshold flipthem: When the winner
  does not need to take all,'' in {\em International Conference on Decision and
  Game Theory for Security}, pp.~74--92, Springer, 2015.

\bibitem{leslie2017multi}
D.~Leslie, C.~Sherfield, and N.~P. Smart, ``Multi-rate threshold flipthem,'' in
  {\em European Symposium on Research in Computer Security}, pp.~174--190,
  Springer, 2017.

\bibitem{pauly1968economics}
M.~V. Pauly, ``The economics of moral hazard: comment,'' {\em The American
  Economic Review}, pp.~531--537, 1968.

\bibitem{shavell1979moral}
S.~Shavell, ``On moral hazard and insurance,'' in {\em Foundations of Insurance
  Economics}, pp.~280--301, Springer, 1979.

\bibitem{gordon2003framework}
L.~A. Gordon, M.~P. Loeb, and T.~Sohail, ``A framework for using insurance for
  cyber-risk management,'' {\em Communications of the ACM}, vol.~46, no.~3,
  pp.~81--85, 2003.

\bibitem{bailey2014mitigating}
L.~Bailey, ``Mitigating moral hazard in cyber-risk insurance,'' {\em JL \&
  Cyber Warfare}, vol.~3, p.~1, 2014.

\bibitem{pal2014improving}
R.~Pal, {\em Improving network security through cyber-insurance}.
\newblock Citeseer, 2014.

\bibitem{pal2017security}
R.~Pal, L.~Golubchik, K.~Psounis, and P.~Hui, ``Security pricing as enabler of
  cyber-insurance a first look at differentiated pricing markets,'' {\em IEEE
  Transactions on Dependable and Secure Computing}, 2017.

\bibitem{ghotbi2012bilevel}
E.~Ghotbi and A.~K. Dhingra, ``A bilevel game theoretic approach to optimum
  design of flywheels,'' {\em Engineering Optimization}, vol.~44, no.~11,
  pp.~1337--1350, 2012.

\bibitem{jenabi2013bi}
M.~Jenabi, S.~M. T.~F. Ghomi, and Y.~Smeers, ``Bi-level game approaches for
  coordination of generation and transmission expansion planning within a
  market environment,'' {\em IEEE Transactions on Power systems}, vol.~28,
  no.~3, pp.~2639--2650, 2013.

\bibitem{zhang2018game}
R.~Zhang and Q.~Zhu, ``A game-theoretic approach to design secure and resilient
  distributed support vector machines,'' {\em IEEE Transactions on Neural
  Networks and Learning Systems}, 2018.

\bibitem{zhang2017bi}
R.~Zhang, Q.~Zhu, and Y.~Hayel, ``A bi-level game approach to attack-aware
  cyber insurance of computer networks,'' {\em IEEE Journal on Selected Areas
  in Communications}, vol.~35, no.~3, pp.~779--794, 2017.

\bibitem{dassios1989martingales}
A.~Dassios and P.~Embrechts, ``Martingales and insurance risk,'' {\em
  Communications in Statistics. Stochastic Models}, vol.~5, no.~2,
  pp.~181--217, 1989.

\bibitem{cizek2005statistical}
P.~Cizek, W.~K. H{\"a}rdle, and R.~Weron, {\em Statistical tools for finance
  and insurance}.
\newblock Springer Science \& Business Media, 2005.

\bibitem{balakrishnan2018exponential}
K.~Balakrishnan, {\em Exponential distribution: theory, methods and
  applications}.
\newblock Routledge, 2018.

\bibitem{engoulou2014vanet}
R.~G. Engoulou, M.~Bella{\"\i}che, S.~Pierre, and A.~Quintero, ``Vanet security
  surveys,'' {\em Computer Communications}, vol.~44, pp.~1--13, 2014.

\bibitem{gantsou2015use}
D.~Gantsou, ``On the use of security analytics for attack detection in
  vehicular ad hoc networks,'' in {\em 2015 International Conference on Cyber
  Security of Smart Cities, Industrial Control System and Communications
  (SSIC)}, pp.~1--6, IEEE, 2015.

\bibitem{zhang2017strategic}
T.~Zhang and Q.~Zhu, ``Strategic defense against deceptive civilian gps
  spoofing of unmanned aerial vehicles,'' in {\em International Conference on
  Decision and Game Theory for Security}, pp.~213--233, Springer, 2017.

\bibitem{ewold1991insurance}
F.~Ewold, ``Insurance and risk,'' {\em The Foucault effect: Studies in
  governmentality}, pp.~197--210, 1991.

\bibitem{peltzman1975effects}
S.~Peltzman, ``The effects of automobile safety regulation,'' {\em Journal of
  political Economy}, vol.~83, no.~4, pp.~677--725, 1975.

\end{thebibliography}
}
\end{document}